\documentclass[11pt, letterpaper]{article}
\usepackage{fullpage}
\usepackage{amsthm}
\usepackage{amsmath,amssymb,amsfonts,nicefrac}
\usepackage{xspace}
\usepackage{color}
\usepackage{url}
\usepackage{hyperref}
\usepackage{bm}
\usepackage{bbm}
\usepackage{times}
\usepackage[labelfont=bf]{caption}

\usepackage{enumitem}

\usepackage{graphicx} %package to manage images
\newtheorem{thm}{Theorem}[section]

\newtheorem{lemma}[thm]{Lemma}

\newtheorem{claim}[thm]{Claim}

\newtheorem{definition}[thm]{Definition}
\newtheorem{remark}[thm]{Remark}

\newtheorem{fact}[thm]{Fact}
\newtheorem{question}{Question}

\newcommand\card[1]{\left| {#1} \right|}
\newcommand\sett[2]{\left\{ \left. #1 \;\right\vert #2 \right\}}

\newcommand\set[1]{{\left\{ #1 \right\}}}
\newcommand\Prob[2]{{\Pr_{#1}\left[ {#2} \right]}}

\newcommand\cProb[3]{{\Pr_{#1}\left[ \left. #3 \;\right\vert #2 \right]}}

\newcommand\Expect[2]{{\mathop{\mathbb{E}}_{#1}\left[ {#2} \right]}}

\newcommand\cExpect[3]{{\mathbb{E}_{#1}\left[ \left. #3 \;\right\vert #2 \right]}}
\newcommand\norm[1]{\| #1 \|}

\newcommand\half{{1\over2}}

\newcommand\inner[2]{\langle{#1},{#2}\rangle}
\newcommand\eps{\varepsilon}

\renewcommand\geq{\geqslant}
\renewcommand\leq{\leqslant}

\newcommand{\rom}[1]{\uppercase\expandafter{\romannumeral #1\relax}}

\usepackage[english]{babel}
\usepackage[square,numbers]{natbib}
\title{
Pandemic Spread in Communities via Random Graphs
}
\date{\vspace{-5ex}}
\author{
Dor Minzer
\thanks{Department of Mathematics, Massachusetts Institute of Technology, Cambridge, USA.}
\and
Yaron Oz
\thanks{Raymond and Beverly Sackler School of Physics and Astronomy, Tel-Aviv University;
School of Natural Sciences, Institute for Advanced Study, Princeton NJ, USA. Supported by ISF Center of Excellence, the IBM Einstein Fellowship and the John and Maureen Hendricks Charitable Foundation at the Institute for Advanced Study.
}
\and Muli Safra
\thanks{School of Computer Science, Tel Aviv University. Supported by the European Research Council (ERC) under the European Union’s Horizon 2020 research and innovation programme (Grant agreement No. 835152), and by an ISF grant and a BSF grant.}
\and Lior Wainstain
\thanks{School of Computer Science, Tel Aviv University, Tel Aviv, Israel.}
}
\begin{document}
\maketitle
% \end{document}
\begin{abstract}
Working in the multi-type Galton-Watson
branching-process framework we analyse the spread of
a pandemic via a general multi-type random contact graph.
Our model consists of several communities,
and takes, as input, parameters that outline the contacts between
individuals in distinct communities.
Given these parameters,
we determine whether there will be an outbreak and if yes, we calculate the size of the giant--connected-component of the graph, thereby, determining
the fraction of the population of each type that would be infected before it ends.
We show that the pandemic spread has a natural evolution direction given by the
Perron-Frobenius eigenvector of a matrix whose
entries encode the average number of individuals of one
type expected to be infected by an individual of another type.
The corresponding eigenvalue is
the basic reproduction number of the pandemic.
We perform numerical simulations that compare homogeneous and heterogeneous spread graphs
and quantify the difference
between them.
We elaborate on the difference between herd immunity and the end of the pandemic and the effect of countermeasures on the fraction of infected
population.

\end{abstract}

\section{Introduction}

Pandemic spread has a tremendous disruptive
impact on the world and consequently a momentous worldwide effort is carried out to figure out the dynamics of the spread and plausible measures that should be taken to control it.
The fundamental question is: Presuming clear knowledge of the distributions of the infectiousness and susceptibility parameters in the population, as well as the correlation between them; how does one calculate the fraction of the population that would contract the disease as a function of time? And when would the spread end with vs. without taking countermeasures?

The spread {\em network} may be viewed as a random combinatorial graph, where vertices correspond to individuals and random edges to  infections.
In real life, this network is not homogeneous but rather heterogeneous, with distinct
individuals and communities being infectious and susceptible to contract the disease to different degrees \cite{review1,review2}.
Indeed, typical data of COVID-19 pandemic around the world
reveals sources of high infection rates, and certain
estimates \cite{Adi} assert
that about $10\%$ of the infected individuals---the superspreaders \cite{NatureCovid2}---cause $80\%$ of the secondary infections, thus much effort was given to understand their influence on the spread of a pandemic (see, e.g.~\cite{dearruda2014role,kitsak2010identification,us2,us1}).

A major source of the complexity of the spread is the structure of the graph
of contacts between individuals and among diverse communities.
On the one hand, it is difficult to obtain precise real data on the spread that would have allowed us to construct that graph in details, while, on the other hand, it is rather complicated
to analyse the properties of such a complex graph.

The aim of our work is to consider a rather general model of heterogeneous random graphs, analyse some of its combinatorial properties and relate them to properties of the spread.
The important feature of our model is that the random graph is multi-type, i.e. it includes
different types of vertices and this allows us to take into account distinct communities
when studying properties of the pandemic outbreak.
The epidemic model that we consider is that of SIR~\cite{SIR} where an infected individual, which will be called a saturated vertex of the graph, cannot be infected
again and hence is removed from the spread process. Within this framework it is customary to consider the size of the infected population at the end of the pandemic 
that corresponds to the giant connected component (abbreviated GCC henceforth) of the graph. This fraction
can be very different from the fraction of the population that is infected until the (effective) reproduction
number drops below $1$ which is conventionally referred to as herd immunity.
We will see this difference explicitly in the numerical simulations.

By employing the multi-type Galton-Watson branching-process framework we are able to analyse a general random contact graph and draw conclusions about the spread of the pandemic. Specifically, given  the parameters defining the distribution of a random graph, we are able to determine whether there would be an outbreak and if so, calculate the size of the GCC of the graph, thereby determining the fraction of population that will contract the disease before it ends, either naturally or as a result of countermeasures.

We will assume $r$ types of communities, wherein community $i$ has size $n_i$ for $i=1,...,r$. We will also assume that we have a symmetric matrix $P$ with non-negative entries, whose
$(i,j)$ entry quantifies the probability that an individual of type $i$ infects an individual of type $j$ as depicted in Figure 1. We will assume that this probabilities are time independent during the natural evolution of the pandemic unless countermeasures are taken at some time.

\begin{figure}[h]
\centering
\begin{tabular}{c}
 \includegraphics[width=90mm]{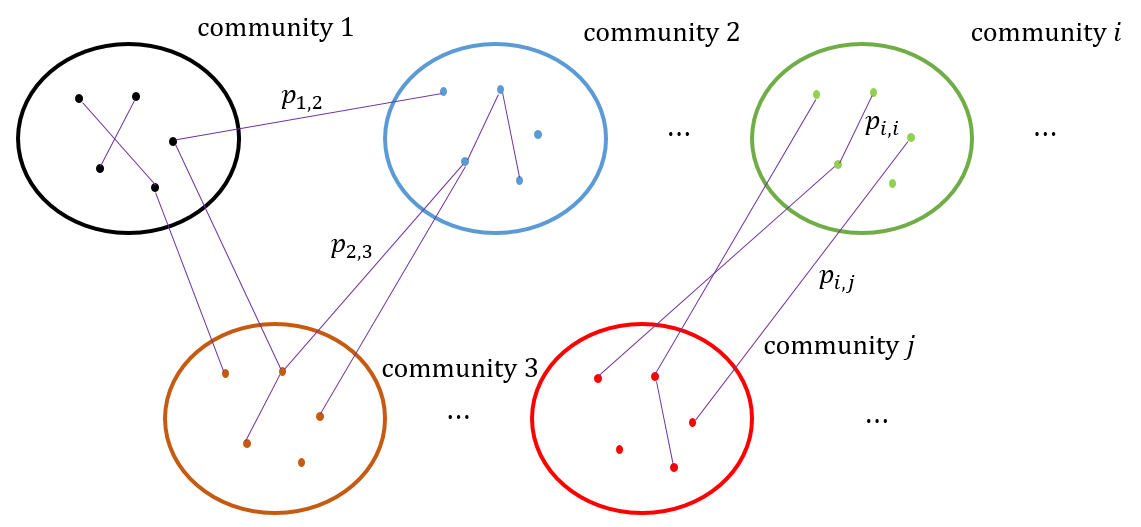} \\[6pt]
\end{tabular}
\caption{A multi-type random graph.}
\label{fig:model}
\end{figure}

This model is motivated by what is often observed as the structure of a society inside
a country or geographical location. Namely, individuals can often be partitioned into
relatively small number of sub-populations (defined by socioeconomic status, religion, size
of household etc.). Within each sub-population, individuals typically behave the same, in terms of internal contacts, as well as contacts with other sub-populations.
This motivates
the view of the contact-network in a society as a graph, wherein each sub-population consists of a set of vertices, and the links between (and inside) these sets of vertices are drawn randomly according to the characteristics of this sub-population.

Our analysis reveals an interesting structure governing the progression of the pandemic.
Despite the random graph potentially having a complicated structure, there is a unique vector that can be associated with the ``direction" of the
spread in the $r$-dimensional space of types, with an associated scalar that corresponds to the basic reproduction number of the pandemic.
More precisely, one defines a matrix $M$ whose $(i,j)$ entry for $i,j=1,...,r$ encodes the average number of individuals of type $j$ expected to be infected by an individual of type $i$.
The matrix $M$ determines the progression of the pandemics and the ``direction" vector is its
Perron-Frobenius eigenvector with {\em the largest eigenvalue $\rho(M)$ corresponding
to the basic reproduction number} of the pandemic.
$\rho(M)$ is the base of the exponential-growth of the disease, before a significant fraction of the population have contracted it.
The components of the Perron-Frobenius eigenvector hold the information about
the potential number of spreaders of each type, hence is intuitively related to the direction of the spread.

Given the input matrix $P$, one can determine when $\rho(M) \leq 1- \eps$, implying the spread ends naturally without an outbreak, in contrast to when $\rho(M) \geq 1 + \eps$, in which case, a large fraction of the population would be infected, namely all the GCC.
In the latter case, we can furthermore perform a calculation of the size of the GCC of the graph.
Moreover, our framework allows us to calculate the fraction of the infected population for each community type separately.

Note that during the evolution of the pandemic the ratios between the numbers of unsaturated vertices (individuals that have not been infected and remain susceptible to contract the disease) from each of the types changes and thus the matrix $M$ changes.
This is because there are types that are more infectious and susceptible than others and they are likely to have a larger percentage of infected vertices.
This implies that the matrix $M$ changes with time.
While this dependency is not important if we are interested only in the final size of the GCC at the end of the pandemic, 
it is important if we are interested in answering questions that are time dependent or before the end of the pandemic.
We will discuss such questions in Section~\ref{sec:simulations}.

\subsection{The graph model}
Our model is best described via the formalism of graph theory.
Let $r\in\mathbb{N}$ be the number of types of individuals in the population, and
let $\vec{n} = (n_1,\ldots, n_r)$ be the vector indicating for each type $i=1,\ldots,r$, the number $n_i$ of individuals of that type. Note, that the numbers $n_i$ are the ones at the beginning of the pandemics.
Let us remark here
that $r$ should be thought of as constant when compared to the total size of the population $n := n_1 + \ldots + n_r$.\footnote{Our analysis nonetheless follows through even if $r$ is a slowly
growing function of the total population, say logarithmic.}

For any two types $i, j\in [r]$ (not necessarily distinct), we have a non-negative parameter $\lambda_{i,j}$ which captures (after an appropriate normalization)
the susceptibility and infection between an individual of type $i$ and an individual of type $j$, this is a term often called {\em transmissibility} (see
for instance \cite{NewmanAlone}).
Let us assume, throughout, that these parameters are symmetric, i.e. $\lambda_{i,j} = \lambda_{j,i}$.
Using these parameters we define the matrix $P = (p_{i,j})_{i,j\in[r]}$ as $p_{i,j} = \lambda_{i,j}/\sqrt{n_in_j}$ which describes the probability that an individual
of type $i$ infects an individual of type $j$ (given that the individual of type $j$ is susceptible).\footnote{When $i=j$, $p_{i,j}=\lambda_{i,i}/n_i$, this is the well-known normalization from the Erd\H{o}s-R\'{e}nyi~\cite{ErdRenT} model $G(n,p)$ where the threshold function was found to be $1/n$ (this model is often stated with parameter $\lambda$ such that $p=\lambda/n$ and the value of $\lambda$ determines whether a giant component will appear or not). Thus, it fits logically for our model, as $\lambda_{i,i}$ describes the strength of connections inside a given community of size $n_i$, i.e. between vertices of the same type. On the other hand, when $i\neq j$, we want $\lambda_{i,j}$ to describe the connection strength between two disjoint sides (each side contains the vertices of one of the types $i$ or $j$). Those edges between different types, resemble the connections between vertices in the bipartite random graph described in~\cite{swedishboy}, where the threshold function for a GCC in $G(m,n;p)$ was found to be $1/\sqrt{mn}$. Therefore, this setting of parameters generalizes both the Erd\H{o}s-R\'{e}nyi and the bipartite random graph models, and will turn out to be suitable for the multi-type graphs, as well.}

Once that matrix is set, let us formulate the distribution $G(\vec{n}, P)$ over graphs $G = (V,E)$ as follows: The vertex set $V$ consists of $r$ disjoint sets of vertices $V_1,\ldots,V_r$, where
$\card{V_i} = n_i$ for all $i\in [r]$; as for the edges, for all $v_i\in V_i$, $u_j\in V_j$, each edge $(v_i,u_j)$ occurs in $G$, independently, with probability $p_{i,j}$.

A random graph sampled this way describes the progression of the pandemic as follows: if at a certain point in time a vertex $v\in V$ was infected, 
then at the next step $v$ infects all vertices adjacent to it in $G$.
In other words, the random variable corresponding to an edge $(v,u)$ occurring in $G$ encapsulates the probability
that these two individuals would be in close proximity when $v$ is infected, as well as the probability that $v$ would pass it along in those interactions (see Figure~\ref{fig:model}).

This perspective highlights the importance of a giant component in the graph $G$ in the study of the size of an outbreak, as well as the number of infected individuals at each step of the spread.
Intuitively, if an individual $v\in V$ is infected, then every individual in the connected component of $v$
would eventually be infected too. %(and furthermore, these are the only occurrences of the disease that can be attributed to $v$).
This motivates us to study the following question:
\begin{question}
  How does the distribution of connected components in $G$ behave as a function of $(\lambda_{i,j})_{i,j\in [r]}$?
\end{question}
Let us remark here that this question is a variant of the well-known problem of determining the size of the giant component in the classical
Erd\H{o}s-R\'{e}nyi model~\cite{AlonSpencer,RandomGraphs}, and closely related variants that have been studied
more recently~\cite{swedishboy,GamarnikMisra,kang2015phase}.

Our method nevertheless---combining spectral considerations and the theory of multi-type Galton-Watson processes---is more natural in our view.
As a biproduct of our approach, we uncover other parameters that may be of interest in the study of the spread of a pandemic.
Our model is a general model of multi-type random graphs. We will describe a branching process over such a random graph in order to analyse some of its combinatorial properties, most importantly, for pandemic analysis, the size of the GCC depending on the initial parameters of the graph. The branching process matches the way a pandemic spreads through society in our model.

Throughout, we pose the basic assumptions of the SIR framework \cite{SIR}. We have two classes of vertices in the graph called
unsaturated and saturated. The unsaturated vertices are individuals susceptible to infection and the saturated vertices are infected individuals.
An individual can only be infected once. At the beginning
of the process all the vertices of the graph are unsaturated except for $1$ (the so called, patient zero), which is infectious. At each step, infectious vertices infect according to their connections with other vertices in the graph, but they infect only other unsaturated vertices. At the next step the vertices that just infected their neighbours, become saturated (i.e. recovered/removed in SIR framework) and step out of the pandemic spread cycle in the sense that they cannot get infected (or infect others) again. Under these assumptions of the SIR framework, the GCC in the random graph corresponds to those individuals who got infected during the process.

\subsection{The descendants matrix and its spectrum}
Consider the matrix $\Lambda\in\mathbb{R}^{r\times r}$ whose $i,j$ entries are equal to $\lambda_{i,j}$ and
recall the matrix $P$ whose $i,j$ entry is $p_{i,j} = \frac{\lambda_{i,j}}{\sqrt{n_i n_j}}$.
Let us further define the matrix $M\in \mathbb{R}^{r\times r}$ as
$M_{i,j} = p_{i,j} n_j=\lambda_{i,j}\cdot\sqrt{\frac{n_j}{n_i}}$. Intuitively, $M_{i,j}$ measures the expected number
of individuals of type $j$ that a given individual of type $i$ would infect.
%It is therefore natural that this matrix is of central importance in the study of the disease, as we demonstrate below.

It is well known (and easy to observe) that $M$ is diagonalizable, and furthermore, its eigenvalues are real (we reproduce a proof of this fact in Section~\ref{sec:prelim}).
The one parameter of utmost interest regarding this matrix is its highest eigenvalue, which will be denoted throughout by $\rho(M)$.
Our result shows that this parameter
is in fact exactly the same as the basic reproduction number (i.e. $R_0$).
Namely, we prove:
\begin{thm}\label{thm:main}
  With the set-up above, for all $r, A\in\mathbb{N}$, $\eps>0$; there exists $N\in\mathbb{N}$ such that the following holds:
  Suppose $n_1,\ldots,n_r$ are such that $n\geq N$, $n_i\leq A n_j$ for all $i,j\in [r]$, and let $\Lambda$, $M$ be the matrices as above.
  Further suppose that $\Lambda$ is connected, and for any $i,j\in [r]$, we either have $\eps\leq \lambda_{i,j}\leq \frac{1}{\eps}$ or $\lambda_{i,j} = 0$.
  Then:
  \begin{enumerate}
    \item if $\rho(M)\leq 1-\eps$, all connected components in $G\sim G(\vec{n}, P)$ are of size
    at most $o(n)$ with probability $1-o(1)$;
    \item if $\rho(M) > 1+\eps$, then there are $\alpha_1,\ldots,\alpha_r > 0$ such that with probability $1-o(1)$,
    $G = (V_1\cup\ldots V_r, E)\sim G(\vec{n}, P)$
    has a connected component $C$ such that $\card{\card{C\cap V_i} - \alpha_i n_i} = o(n)$ for all $i\in[r]$, and all other connected
    components have size $\le{\sf polylog}(n)$.
  \end{enumerate}
\end{thm}
In our view, the main case of interest is $\rho(M) > 1+\eps$ (otherwise the pandemic disappears fairly quickly with little impact).
In this case, we give explicit, simple equations for computing $\alpha_1,\ldots,\alpha_r$ (see Equation~\eqref{eq:gen_fn_def} and Theorem~\ref{thm:extinct_solve_gen}).
Furthermore, our proof demonstrates that the parameter $\rho(M)$ behaves like the basic reproduction number in a manner more general than just
implying whether the disease is likely to spread out or not. In particular, our arguments show that the growth in
the number of infections over time is an exponential function whose base is $\rho(M)$ (at least as long as the total number of infections has not reached a significant
proportion of the population).

\subsection{Comparison to related literature}

A large number of approaches have been proposed to model pandemic spread, as well as other information transmission processes, using percolation in complex networks. Thus, much effort was devoted to characterizing sharp thresholds for many such models. In~\cite{Karrer_2014}, Karrer, Newman and Zdeborova consider a related problem of percolation over sparse networks. They address the random subgraph model over a general host graph, wherein one starts with some host graph $H$ and samples a random subgraph $G$ of $H$ by independently including each edge from $H$ with probability $p$. Their work considers the problem of locating the critical probability $p_c$, such that for $p<p_c$ the graph $G$ is unlikely to contain large clusters (i.e. large connected components), and for $p>p_c$ the graph is likely to contain large connected components. The authors relate this critical probability to the non-backtracking random walk matrix on $H$, and in particular to the topmost eigenvalue of a matrix related to the adjacency matrix of $H$. The intuition behind this result is that the topmost eigenvalue represents a certain notion of average degree in the percolation graph, so that, in a sense, if the topmost eigenvalue is $\lambda$, then a vertex in $G$ has $\lambda$ neighbours on average, which results in an exponential growth with the number of steps of the process. The authors argue that their result holds for large girth, locally tree-like host graphs, and support their results by performing numerical simulations over several host graphs. Furthermore, a similar result was obtained by Bollob\'{a}s et al.~\cite{BoloDense} for dense networks. In the latter, the threshold function obtained for $p$ is the inverse of the leading eigenvalue of the simple adjacency matrix.

\noindent Our work has several points of similarity and dissimilarity with~\cite{Karrer_2014,BoloDense}, which we outline next.
\begin{enumerate}
\item Our work seeks formal mathematical proof of our results, and we are therefore constrained to work with a less general class of host graphs. Indeed, there are pathological examples of host graphs and to allow for a formal proof one has to make relatively strong assumptions on the structure of the host graph. In our case, we are working with blow-ups of graphs on $k$-vertices with $k$ being constant, meaning that we start with any host graph $H$ over $k$ vertices and replace each vertex in it with a new cloud of vertices. To simplify terminology, we refer to the original vertices of $H$ as types, and denote the cloud of vertices in our graph that replaces type $i$ in $H$ by $V_i$. Our graph only allows edges between a vertex in $V_i$ and a vertex in $V_j$ if the original host graph $H$ had an edge between types $i$ and $j$.

\item Our model is more general in the sense that it allows different
survival probabilities for different types of edges (occupation probabilities in the language of~\cite{Karrer_2014}). To be more precise, for any two types $i,j\in [k]$, the probability of an edge $(v_i,v_j)$ where $v_i\in V_i$ and $v_j\in V_j$ (namely, in the language of~\cite{Newman},
the occupation probabilities of edges coming from different types may be different). These probabilities are the input to our problem (related to the parameters $\lambda_{i,j}$ and the initial population sizes), and we wish to determine when such graphs are likely to contain large clusters, and when they are unlikely.

\item Our work also identifies that the important parameter in our problem is the topmost eigenvalue, but of a different matrix. The matrix we consider here is $M$, a $k$ by $k$ associated with the prototype of our host graph, whose entries are appropriate normalizations of the susceptibility and rate of interaction between different types. As in~\cite{Karrer_2014}, this parameter also carries with it a certain notion of average-degreeness, however it is more subtle in our case as not all of the neighbours contribute equally to further percolation of the process. Still, we prove that the magnitude of the topmost eigenvalue of the matrix $M$ determines the existence of large clusters in our model.
\end{enumerate}

Pandemic spread  models of networks with general power-law degree distributions were
studied by Newman in~\cite{NewmanAlone} and later by Miller in~\cite{PhysRevE.76.010101}.\footnote{Constructions of other scale-free networks with power-law degree distributions appear in~\cite{Albert_2002} by Albert and Barab\'{a}si, which were found to be useful in modeling many different connection-networks, such as the world wide web.}
The transmissibility between each two vertices $u,v$ is drawn according to an appropriate random variable $T_{u,v}$. All those random variables are i.i.d. and taken from a predefined distribution, thus a new parameter $T$ is introduced which takes a weighted average of these variables, according to the predefined distribution. In this model, vertices can be sometimes infectious for multiple rounds, or rather, be non-infectious (albeit being carriers of the virus) for multiple rounds.
The paper~\cite{NewmanAlone} develops several methods in order to analyze questions concerning parameters such as: 
the relation between the degree of a vertex and its probability to avoid infection; sizes of outbreaks; and finally thresholds for the occurrence of an outbreak. 
Those are obtained using the generating functions method~\cite{Newman}, described also in~\cite[Section IV.C]{review2} and initially in~\cite{1990G}.

The general multi-type framework that we analyse in this paper was described in~\cite{Newman_2003} as an undirected network whose vertices are partitioned into types that interact with each other. A study of such model using the method of generating functions with an application to pandemic spread has been carried out in~\cite{Allard_2009} where, under some assumptions, a criterion for the threshold to having GCC was proposed. The criterion differs from ours
and since the proof methods used in \cite{Allard_2009} are non-rigorous it is not clear to us when their result holds.  In the current work we study a model where the degree distributions of the vertices are binomial distributions and prove  rigorously the results about the GCC. The advantage of having a rigorous proof is two-fold: first, it explains exactly under which conditions on the parameters the result has to hold; second, in order to make the proof go through we identify several parameters that then may actually become the center of interest, and are completely absent in~\cite{Allard_2009}. In this case, as our technique mostly relies on spectral graph theory, we identify the types matrix $M$ and the importance of its topmost eigenvalue $\rho(M)$, which for pandemics has the interpretation as the basic reproduction number, but more generally may be regarded as the rate of increase of the components as a function of "time" (where each step of the aforementioned Galton-Watson branching process advances the "time" by exactly one step). The associated eigenvector also plays a crucial role in our proof, and turns out to encode the number of "active" vertices at various stages of the process (a term we later refer to as "unsaturated vertices"). In contrast, in~\cite{Vazquez_2006} a topmost eigenvalue of a matrix $\widetilde{R}$ is indeed noted to be related to the basic reproduction number of the disease. However, no explicit statement is given to relate the eigenvalue to the emergence of the giant component or to its growth rate through a branching process on the graph, nor a connection is drawn between the direction of the spread and the appropriate leading eigenvector, which was one of the main goals of this paper.

More recent analysis of pandemic spread in networks include~\cite{Herd1,He} and~\cite{us1,us2}, where the basic graph is complete, and the heterogeneity is received by drawing connections according to some specific power-law degree distribution (e.g. appropriate Gamma distributions). Using these models, herd immunity factors are calculated for families of distributions, in particular for COVID-19. Moreover, a numerical simulation of multi-type heterogeneous society network where described in~\cite{science}, which resembles our model in the sense that it considers partitioning society according to different characteristics, but takes a rather experimental approach, to tackle specifically COVID-19 data. In summary, those described models often assume the homogeneity of the society (individual vertices are indistinguishable) and they obtain heterogeneity only by applying degree distributions that describe the disease parameters. Others, allow multiple communities, but include only simulations of specific data, and not general analytical results.

\subsection{Related applications}

The questions addressed here are applicable beyond the spread of pandemics.
Taking a large set of vertices and analyzing their interaction by first partitioning them into types and then assuming random connections, takes place in many distinct areas
where the setup consists of a large complex network. These include networks in the physical world such as in life sciences, in the virtual world of the internet, and in the society
(for a survey see e.g. \cite{kleinberg2006complex}).
One natural example is search algorithms.
The web is a graph whose vertices are all the known web-pages and where a directed edge connects one vertex to another if the first includes a link to the other.
Trying to figure out the best pages for a given query amounts to analyzing that graph.
The assumption that web-pages have types, and that connections between types can be assumed to be quite random, have been promoted both theoretically (for a review see e.g. \cite{klein}), as well as in practice, where search algorithms that run in practice are rumored to employ such tactics.

\vskip1cm
\paragraph{Organization:}
The paper is organized as follows.
In Section~\ref{sec:prelim} we introduce some preliminaries and probabilistic tools
and outline the multi-type Galton-Watson
branching process framework.
In Section~\ref{sec:GW_properties_more} we prove various properties of the multi-type Galton-Watson process
including a quantitative bound on the extinction probabilities, as well as
the Galton-Watson
simulation of the connected components in random graphs.
%and the distribution of the population size at intermediate steps.
In Section~\ref{sec:proff_main} we prove Theorem~\ref{thm:main}.
%We consider the sub-critical case and the super-critical case.
In Section~\ref{sec:simulations} we perform numerical simulations to demonstrate the analytical statements.
Section~\ref{sec:discussion} is devoted to a discussion and outlook.

\section{Preliminaries}\label{sec:prelim}
\paragraph{Notations.}
Throughout the paper we use asymptotic big-$O$ notation. We write $X = O(Y)$ or $Y = \Omega(X)$ to say that there exists an absolute
constant $C$ such that $X\leq C\cdot Y$. Sometimes, this constant will depend on other parameters, say on $\eps$,
in which case we write $X = O_{\eps}(Y)$ or $Y = \Omega_{\eps}(X)$. We let $[r] = \set{1,\ldots,r}$ and $\mathbb{N} = \set{0,1,2,\ldots}$.
We use the standard Euclidean inner product $\inner{x}{y} = \sum\limits_{i} x_iy_i$ and $\ell_p$-norms
$\norm{x}_p^p = \sum\limits_{i} \card{x_i}^p$.

\begin{definition}
  We say $n\in\mathbb{N}^r$ is $A$-balanced if for any $i,j\in [r]$ it holds that $n_i\leq An_j$.
\end{definition}
\begin{definition}
  We say $\Lambda = (\lambda_{i,j})_{i,j\in[r]}$ is $\eps$-separated if for any $i,j\in [r]$ it holds that $\lambda_{i,j} = 0$ or $\eps\leq \lambda_{i,j}\leq \frac{1}{\eps}$.
\end{definition}
\begin{definition}
  We say a symmetric $\Lambda = (\lambda_{i,j})_{i,j\in[r]}$ is connected if the graph whose vertices are $[r]$ and there is an edge between $i$ and $j$
  if $\lambda_{i,j}>0$, is connected.
\end{definition}

\subsection{Spectral properties of the matrix $M$}
The following lemma asserts  that our matrix $M$ is diagonalizable, and furthermore the eigenvalues and eigenvectors satisfy several useful properties.
\begin{lemma}\label{lemma:eigenbasis}
  For all $r\in\mathbb{N}$, $\eps>0$ there exists $\delta>0$ such that the following holds.
  Suppose $\Lambda$ is $\eps$-separated and connected. Then there exists a basis $x_1,\ldots,x_r$ of $\mathbb{R}^r$ consisting of eigenvectors of
  $M$. Furthermore, if $\theta_1 \geq \theta_2\geq\ldots\geq \theta_r$ are the eigenvalues of $M$,
  %corresponding to $x_1,\ldots,x_r$ respectively,
  then
  %all entries of $x_1$ are at least $\delta$, and
  $\card{\theta_i} \leq \theta_1-\delta$ for all $i\neq 1$.
\end{lemma}
\begin{proof}
Let $D^{-\half} = {\sf diag}(n_1^{-\half},\ldots,n_r^{-\half})$, and observe that $M = P D$.
Consider the matrix $\Lambda = D^{\half} P D^{\half} = (\lambda_{i,j})_{i,j}$. Since it is
symmetric, there is a basis $y_1,\ldots, y_r$ consisting of eigenvectors of it with
real eigenvalue $\theta_1,\ldots,\theta_r$.
Thus, for each $i\in [r]$,
\[
M (D^{-\half} y_i)
=
P D (D^{-\half} y_i) = P D^{\half} y_i = D^{-\half} (D^{\half} P D^{\half} y_i) = \theta_i D^{-\half} y_i,
\]
so $x_i = D^{-\half} y_i$ is an eigenvector of $MD$ with eigenvalue $\theta_i$.

For the furthermore statement, we use a quantitative version of the Perron-Frobenius theorem. More precisely,
we use \cite[Inequality (3)]{PerronQuantitative}, and for that we first bound the quantity $m(A)$ defined therein.

Normalize $y_1$ so that $\norm{y_1}_2 = 1$; we show that $y_1(i) \geq \delta$ for all $i$, for some $\delta>0$ depending only on $r$ and $\eps$.
Indeed, first taking $i$ that maximizes $y_1(i)$, we have that $y_1(i)\geq r^{-\half}$. For the rest of the types, we first upper bound $\rho(\Lambda)$:
\begin{equation}\label{eq:rho_upper_bd}
\rho(\Lambda) = \frac{(\Lambda y_1)(i)}{y_1(i)} = \frac{\sum\limits_{j=1}^{r} \lambda_{i,j} y_1(j)}{y_1(i)}
\leq \frac{y_1(i)\sum\limits_{j=1}^r\lambda_{i,j}}{y_1(i)}
\leq \frac{r}{\eps}.
\end{equation}
We now lower bound $y_i(j)$ for all $j$. Take $j$ minimizing $y_1(j)$; as $\Lambda$ is connected, there is a path of length $\ell\leq r$ from $j$ to $i$, so
\[
\rho(\Lambda)^\ell y_1(j) = (\Lambda^{\ell} y_1)_j \geq \sum\limits_{j_1=j,\ldots,j_\ell=i}{\lambda_{i,j_1}\lambda_{j_1,j_2}\cdots \lambda_{j_{\ell-1},j_{\ell}} y_{1}(j_{\ell})}\geq
\eps^{\ell}y_1(i)
\geq \frac{\eps^{\ell}}{\sqrt{r}},
\]
so
\begin{equation}\label{eq1}
y_1(j)\geq \frac{\eps^{2r}}{\sqrt{r} \cdot r^r}.
\end{equation}

We can now lower bound $m(\Lambda)$ from \cite[Inequality (3)]{PerronQuantitative} as:
\[
m(\Lambda) = \frac{r}{\norm{y_1}_2^2}\min_{F\subsetneq[r], F\neq \emptyset}\sum\limits_{i\in F, j\not\in F}{\lambda_{i,j} y_1(i) y_1(j)}.
\]
As $\Lambda$ is connected, there are $i\in M$, $j\not\in M$ such that $\lambda_{i,j}>0$ so that $\lambda_{i,j}\geq \eps$, and then
we get $m(A)\geq r \eps \min_{k}y_1(k)^2 \geq r\eps\cdot \frac{\eps^{4r}}{r^{2r+1}} = \delta(r,\eps)>0$. Therefore, by \cite[Inequality (3)]{PerronQuantitative}
$\card{\theta_i}\leq \theta_1 - \frac{4}{r(r-1)}\delta$ for all $i\neq 1$, and we're done.
\end{proof}

Let $x_1,\ldots,x_r$ be the basis from Lemma~\ref{lemma:eigenbasis}. Since this is a basis, any vector $\vec{z}$ may be written as a linear
combination of these vectors. However, since this is not an orthonormal basis, finding these coefficients and proving estimates on them may be
a bit tricky. In the lemma below, we use the fact that this basis is a multiplication of an orthonormal basis with a diagonal matrix,
to prove several estimates on such coefficients that will show up in our proofs.
\begin{lemma}\label{lemma:decompose_std}
  For all $A,r\in\mathbb{N}$ and $\eps>0$ there are $\delta>0$ and $S>0$ such that the following holds.
  Let $\vec{z}\in [0,\infty)^r$, and write $\vec{s} = \sum\limits_{j}\alpha_j x_j$.
  Then:
  \begin{enumerate}
    \item $\norm{\alpha_1 x_1}_2\geq \delta \norm{\vec{z}}_2$;
    \item for all $j$, $\norm{\alpha_j x_j}_2\leq S\norm{\vec{z}}_2$.
  \end{enumerate}
\end{lemma}
\begin{proof}
  As $y_1,\ldots,y_r$ is orthonormal, we have $D^{\half} \vec{z} = \sum\limits_{j=1}^{r}\alpha_{j} y_j$ where
  $\alpha_j = \inner{D^{\half} \vec{z}}{y_j}$, and so $\vec{z} = \sum\limits_{j=1}^{r}\alpha_{j} x_j$. Note that by~\eqref{eq1}
  we have
  \[
  \alpha_1 \geq \min_{j}\sqrt{n_j}y_1(j)\sum\limits_{k=1}^{r} z_k\geq \min_{j}\sqrt{n_j}y_1(j)\max_k z_k\geq \min_j\sqrt{n_j}\frac{\eps^{\ell}}{\sqrt{r}} \frac{\norm{z}_2}{r}.
  \]
  Thus,
  \[
  \norm{\alpha_1 x_1}_2 = \alpha_1 \norm{D^{-\half} y_1}_2\geq \frac{\alpha_1}{\max_j\sqrt{n_j}}\norm{y_1}_2
  \geq\frac{\min_j\sqrt{n_j}\norm{z}_2}{\max_j\sqrt{n_j}} \frac{\eps^{\ell}}{r^{3/2}} \geq \frac{1}{\sqrt{A}} \frac{\eps^{\ell}}{r^{3/2}}\norm{z}_2 \geq \delta(A,r,\eps)\norm{z}_2.
  \]
  This proves the first item. For the second item, fix $j$ and note that
  \[
  \alpha_j = \inner{D^{\half}\vec{z}}{y_j}\leq \norm{D^{\half} \vec{z}}_2\norm{y_j}_2\leq \max_{j}\sqrt{n_j}\norm{\vec{z}}_2.
  \]
  Therefore,
  \[
  \norm{\alpha_j x_j}_2
  \leq\max_{j}\sqrt{n_j}\norm{\vec{z}}_2\norm{D^{-\half} y_j}_2
  =\max_{j}\sqrt{n_j}\norm{\vec{z}}_2 \sqrt{\sum\limits_{k=1}^{r} \frac{1}{n_k} y_j(k)^2}
  \leq \frac{\max_{j}\sqrt{n_j}\norm{\vec{z}}_2}{\min_k \sqrt{n_k}} \norm{y_j}_2,
  \]
  which is at most $\sqrt{A}\norm{\vec{z}}_2$.
\end{proof}

\subsection{The mutli-type Galton-Watson process}\label{sec:GW_basics}
A key component in our analysis is the multi-type Galton-Watson process.
In this section we introduce this process
and recall several of its basic properties.
In Section~\ref{sec:GW_properties_more} we state and prove additional properties
of it that will be important in the proof of our main results. We will specialize our exposition to our case of interest,
in which the number of offsprings is distributed according to a binomial distribution. We refer the reader to~\cite{Harris} for a more systematic treatment
of this process.

\subsubsection*{Development of the Galton-Watson Process}
The parameters defining the Galton-Watson process are identical to the parameters defining our disease: namely a
vector $\vec{n} = (n_1,\ldots,n_r)$ of integers, and a matrix $P := (p_{i,j})_{i,j\in [r]}$
positive entries. The process starts from an initial configuration $x\in\mathbb{N}^r$, specifying for each type
$i\in [r]$ the number of individuals of type $i$, that is, $x(i)$. Individuals are classified as either ``unsaturated'', if
we have not explored their offsprings yet, and otherwise are classified as ``saturated''. Initially, all individuals are classified
as unsaturated. There are two equivalent ways to develop this process, both of which will be useful for us:

\begin{enumerate}
  \item {\bf Sequential Galton-Watson process.}
  At each step, as long as there is an unsaturated individual, the process picks up some unsaturated individual, $w$, and explores its offsprings. Namely, suppose $w$ is of type $i$, then for each $j\in [r]$, the process
  generates independent Binomial samples $N_j\sim {\sf Binomial}(n_j, p_{i,j})$. The process
  adds $N_j$ unsaturated individuals of type $j$ for all $j\in [r]$, and changes the classification of $w$ to ``saturated''. The process halts
  when all individuals are classified as saturated.
  \item {\bf Parallel Galton-Watson process.} Here, at each point in time instead of picking a single unsaturated individual from the
  list, we look at them all together, and generate their descendants simultaneously.
\end{enumerate}

We denote by ${\sf pop}(\text{GW}(\vec{n},P,x))$ the random variable measuring the total population in the process
once it terminates, when starting with initial population $x\in\mathbb{N}^{r}$ (defining it as $\infty$ if the process does't terminate).

\subsubsection*{The extinction probabilities}
One aspect of the Galton-Watson process we will be concerned with are the extinction probabilities with a given initial
configuration, as well as the distribution of the number of individuals in case the process terminated. In this case as well,
it makes sense to consider the matrix $M = (p_{i,j} n_j)_{i,j\in [r]}$. We record below two well-known facts, showing that in
this case the parameter that determines whether the Galton-Watson process will terminate is $\rho(M)$.
\begin{fact}\label{fact:GW_subcript}
  Suppose $P$ is connected and $\rho(M)\leq 1-\eps$, then $\Prob{}{\text{GW}(\vec{n},P,x)\text{ goes extinct}} = 1$ for all $x$.
\end{fact}
\begin{proof}
  Let $x$ be the initial configuration of the Galton-Watson process, and define $x_0 = x$, $x_{t+1} = M x_t$ for each $t\geq 0$.
  Since by Lemma~\ref{lemma:eigenbasis} all eigenvalues of $M$ are at most $\rho(M)$ in absolute value, it follows that
  $\norm{x_{t+1}}_2\leq \rho(M) \norm{x_t}_2$ for all $t$, and by induction $\norm{x_j}_2\leq \rho(M)^j\norm{x}_2$.

  We take the parallel view of the development of the Galton-Watson process.
  Let $X_t$ be the number of individuals explored at time $t$. Note that the expectation of $X_t$ is
  $\norm{x_t}_1\leq \norm{x_t}_2$, so by Markov's inequality
  \[
   \Prob{}{{\sf pop}(\text{GW}(\vec{n},\Lambda))\geq k}
   = \Prob{}{\sum\limits_{t=1}^{\infty} X_t \geq k}
   \leq \frac{\Expect{}{\sum\limits_{t=1}^{\infty} X_t}}{k}
   \leq \frac{\sum\limits_{t=1}^{\infty} \norm{x_t}_2}{k}
   \leq \frac{\sum\limits_{t=1}^{\infty} \rho(M)^t \norm{x}_2}{k},
  \]
  and using the sum of geometric series this is at most
  $\frac{\rho(M)\norm{x}_2}{(1-\rho(M))k}$. Sending $k$ to infinity finishes the proof.
\end{proof}
\begin{remark}
  The above argument can actually give us strong bounds on the total size of the population when the Galton-Watson process terminates.
  Indeed, in Section~\ref{sec:sub_crit} we give an adaptation of this argument to handle the sub-critical case in the graph process.
\end{remark}

Next, it is natural to ask what can be said about the extinction probabilities in the super critical case, $\rho(M)\geq 1+\eps$.
For this, we appeal to a result from~\cite{Csernica}, stating that these extinction probabilities satisfy a simple
system of equations. Towards this end, consider for each type $i$ the probability generating function
of the distribution of descendants of an individual of type $i$:
\begin{equation}\label{eq:gen_fn_def}
  f_i(z_1,\ldots,z_r)
  =
  \sum\limits_{k_1,\ldots,k_r}\prod\limits_{\ell=1}^{r}\Prob{}{{\sf Binomial}(n_\ell, p_{i,\ell}) = k_{\ell}} z_{\ell}^{k_{\ell}}.
\end{equation}
In our case of interest, as the number of descendants of $i$ from different types are independent random variables, we may
simplify the about equation and simply write
\begin{equation}\label{eq:gen_fn_def_simplified}
  f_i(z_1,\ldots,z_r)
  = \prod_{\ell=1}^{r}{\sf Gen}_{n_\ell,p_{i,\ell}}(z_\ell),
\end{equation}
where ${\sf Gen}_{n,p}(z)$ is the generating function of ${\sf Binomial}(n,p)$. With these notations, we define
$f(\vec{z}) = (f_1(\vec{z}),\ldots,f_r(\vec{z}))$. In this language,~\cite[Theorem 7.4]{Csernica} reads:
\begin{thm}\label{thm:extinct_solve_gen}
  Suppose $\rho(M)\geq 1+\eps$, and let $\alpha_i$ be the extinction probability of
  $\text{GW}(\vec{n}, P, e_i)$. Then the vector $\vec{\alpha} = (\alpha_1,\ldots,\alpha_r)$ is a
  fixed point of $f$, i.e. $f(\vec{\alpha}) = \vec{\alpha}$, and $\alpha_i < 1$ for all $i$. Moreover, for any other non-trivial fixed point $\vec{\beta}$, $\alpha_i\leq \beta_i$ for all $i\in[r]$.
\end{thm}

\begin{remark}\label{remark:gcc-equation}
For or most purposes, it is often
useful to approximate the moment generating function of a binomial random variable as
${\sf Gen}_{n,p}(z)\approx e^{-np(1-z)}$, thereby getting $f_i(z_1,\ldots,z_r)\approx e^{-(M (\vec{1}-\vec{z}))_i}$.
\end{remark}

\subsection{Probabilistic tools}
We will need the notion of stochastic domination as well as the following simple fact in our proof.
\begin{definition}[Stochastic Domination]
Let $X, Y$ be two real-valued random variables. We say $X$ stochastically dominates $Y$ if for
all $t\in \mathbb{R}$ it holds that $\Prob{}{X\geq t}\geq \Prob{}{Y\geq t}$.
\end{definition}

\begin{fact}\label{fact:domination_coupling}
Let $X, Y$ be two real-valued random variables. Then $X$ stochastically dominates $Y$ if and only
if there is a coupling $(\widehat{X}, \widehat{Y})$ of $X$ and $Y$, such that $\widehat{X}\geq \widehat{Y}$ always.
\end{fact}

\section{Properties of multi-type Galton-Watson processes}\label{sec:GW_properties_more}
In this section, we prove several more properties of the multi-type Galton-Watson process.

\subsection{A quantitative bound on the extinction probabilities}
Suppose $x\in\mathbb{N}^r$ is the initial configuration of a Galton-Watson process.
%(i.e.\, $x_i$ is the number of individuals there are of type $i$ in the beginning).
A quantity that we'll often be interested in is
$Mx$, which measures the ``expected configuration'' of the population after a single step of development
in parallel. That is, $(Mx)_i$ is the expected number of individuals of type $i$ after a single step of
development in parallel. Thus, studying norms of such expressions, or more generally of expressions such
as $M^{k} x$, is important in understanding the Galton-Watson process.
%and for that we prove the following claim.
\begin{claim}\label{claim:pump}
  For all $A,r\in\mathbb{N}$ and $\eps>0$, there exist $k\in\mathbb{N}$ and $\delta>0$, such that if
  $\vec{n}$ is $A$-balanced, $\Lambda$ is connected and $\eps$-separated and $\rho(M)\geq 1+\eps$,
  then for any vector $\vec{z}$ with non-negative entries it holds that:
  \begin{enumerate}
    \item $\norm{M^{k}\vec{z}}_2\geq 2\norm{\vec{z}}_2$;
    \item for all $t\geq k$,
    \[
    \norm{M^{t}\vec{z}}_2\geq \delta \rho(M)^t\norm{\vec{z}}_2\qquad\qquad
    \norm{M^{t}\vec{z}}_2\leq \rho(M)^t\norm{\vec{z}}_2.
    \]
  \end{enumerate}
\end{claim}
\begin{proof}
  We begin with the first item.
  Write $\vec{z} = \sum\limits_{j=1}^{r} \alpha_j x_j$, and note that for all $k$ we have
  $M^{k}\vec{z} = \sum\limits_{j=1}^{r} \theta_j^k\alpha_j x_j$, where $\theta_j$ is the eigenvalue corresponding to
  $x_j$. Thus,
  \[
  \norm{M^{k}\vec{z}}_2\geq \norm{\theta_1^k\alpha_1 x_1}_2 -
  \max_{\ell\neq 1} \card{\theta_{\ell}}^k \sum\limits_{j=1}^{r} \norm{\alpha_j x_j}_2
  \geq \left(\rho(M)^k \delta_1 - (\rho(M) - \delta_2)^k r S\right)\norm{\vec{z}}_2,
  \]
  where $\delta_1, S>0$ are from Lemma~\ref{lemma:decompose_std}, and $\delta_2$ is from Lemma~\ref{lemma:eigenbasis}.
  As $\rho(M)\geq 1+\eps$, there is $k$ depending only on $\delta_1,\delta_2,r,S$ (and therefore only on $A,r,\eps$) such that
  $\rho(M)^k \delta_1 - (\rho(M) - \delta_2)^k r S\geq 2$, completing the proof.

  For the second item, we may find large enough $k$ such that $\rho(M)^t \delta_1 - (\rho(M) - \delta_2)^t r S\geq \half\delta_1 \rho(M)^t$
  for all $t\geq k$, which gives the first inequality of the second item. The second inequality follows immediately from definition of $\rho(M)$.
\end{proof}

Next, we show that in our set-up of the Galton-Watson process,
in the super-critical case $\rho(M)\geq 1+\eps$ the probability that the process never goes extinct
is bounded away from $0$.
\begin{lemma}\label{lem:GW_supcrit}
  For all $\eps>0$, $A,r\in\mathbb{N}$ there exists $\delta>0$ and $N\in\mathbb{N}$ such that the following holds.
  Suppose $\vec{n}$ is $A$-balanced, $n\geq N$, $\Lambda$ is connected and $\eps$-separated and
  and $\rho(M)\geq 1+\eps$. Then for all $i\in [r]$,
  \[
  \Prob{}{\text{GW}(\vec{n},P,e_i)\text{ goes extinct}} \leq 1-\delta.
  \]
\end{lemma}
\begin{proof}
Recall the functions $f_1(\vec{z}),\ldots,f_r(\vec{z})$ and $f = (f_1(\vec{z}),\ldots,f_r(\vec{z}))$
from Section~\ref{sec:GW_basics}, and let $\alpha_i$ be the extinction probability of $\text{GW}(\vec{n},\Lambda,e_i)$.

By Theorem~\ref{thm:extinct_solve_gen}, $\alpha_1,\ldots,\alpha_r$ satisfies $\alpha_i = f_i(\vec{\alpha})$ and $\alpha_i<1$ for all $i$.
Our main goal is to show that under the conditions of the lemma, there is $i$ such that $z_i\geq 1-\delta$, where $\delta = \delta(A,r,\eps)>0$;
once we show that, it will be easy to deduce this is in fact the case for all $i$.
The proof is very similar to the proof of~\cite[Theorem 7.4]{Csernica}, except that we replace a part of the argument therein with
Claim~\ref{claim:pump} above.

Specifically, let $\delta>0$ to be chosen later, and assume towards contradiction that $\alpha_i\geq 1-\delta$ for all $i$.
Let $\vec{z} = 1-\vec{\alpha}$. For large enough $n = n_1+\ldots+n_r$, we may write
\[
f_i(\vec{\alpha}) = e^{-(M\vec{z})_i} + o(\norm{z}_2)
\leq 1-(M\vec{z})_i + o(\norm{z}_2) + O(\norm{\vec{M z}}_2^2)
\leq 1-(M\vec{z})_i + o(\norm{z}_2) + O(\norm{\vec{z}}_2^2),
\]
so we get that $f(\vec{1}-\vec{z}) = \vec{1} - M\vec{z} + \zeta$ where $\norm{\zeta}_2\leq O(\norm{\vec{s}}_2^2) + o(\norm{s}_2)$.
Therefore, applying this twice, we get
$f(f(\vec{1}-\vec{z})) = f(\vec{1} - M\vec{z} + \zeta) = M(M\vec{z} - \zeta) + \zeta' = \vec{1} - M^2\vec{z} + \zeta''$ where
$\norm{\zeta''}_2 \leq O(\norm{\vec{s}}_2^2) + o(\norm{s}_2)$.

Take $k$ from Claim~\ref{claim:pump}. By induction, we get that
\[
\underbrace{f\circ\ldots\circ f}_{k\text{-times}}(\vec{\alpha})
=\vec{1} - M^k\vec{z} + \zeta,
\]
where $\norm{\zeta}_2\leq  O_k(\norm{\vec{z}}_2^2) + o_k(\norm{z}_2)\leq C\sqrt{r\delta}\norm{\vec{z}}_2$ for large enough $n$ and some $C = C_{A,r,k,\eps}>0$. Now
as $\alpha$ is a fixed point of $f$, we get $\vec{\alpha} = \underbrace{f\circ\ldots\circ f}_{k\text{-times}}(\vec{\alpha})$, and combining
the last two equations we get $\zeta = M^k\vec{z} - \vec{z}$. Thus,
\[
C\sqrt{r\delta}\norm{\vec{z}}_2
\geq
\norm{\zeta}_2 \geq \norm{M^k\vec{z}}_2 - \norm{\vec{z}}_2
\geq \norm{\vec{z}}_2,
\]
giving a contradiction if $\delta < \frac{1}{C^2 r}$.

We conclude that there is $i$ such that $\alpha_i\leq 1-\delta$ for $\delta = \frac{1}{2C^2 r}$, and we next use the connectedness of $\Lambda$
to argue that $\alpha_j\leq 1-\delta'$ for some $\delta'(A,r,\eps)>0$. Indeed, we note that for each $j\in [r]$, as there is a path from $j$ to
$i$ of length $\ell\leq r$, and as $\Lambda$ is $\eps$-separated and $\vec{n}$ is $A$-balanced we get that the probability that
$\text{GW}(\vec{n}, \Lambda, e_j)$ produces an individual of type $i$ after $\ell$ steps is at least $\xi(A,\eps,r)>0$.
Therefore, the probability that this process survives is at least $\xi(A,\eps,r)$ times the probability that
$\text{GW}(\vec{n}, \Lambda, e_i)$ survives, which is at least $\delta$. Thus, the claim is proved for $\delta' = \xi(A,\eps,r)\delta$.
\end{proof}

\subsection{The Galton-Watson simulation of connected components in random graphs}
The following lemma will be useful for us to analyze the size of connected components in graphs
by translating such questions to their analogous counterparts in the Galton-Watson setting.
\begin{lemma}\label{GWSandwich}
Let $k\in \mathbb{N}$, $\vec{T},\vec{n}\in\mathbb{N}^r$ be such that $k\leq \vec{T}(i)\leq n_i$ for all $i$.
Let $i\in [r]$ be a type and $v_i\in V_i$ be a vertex in $G$ of type $i$. Then,
\[
\Prob{}{{\sf pop}(\text{GW}(\vec{n}-\vec{T},P,e_i))\geq k}
\leq
\Prob{}{\card{{\sf CC}(v_i)}\geq k}
\leq
\Prob{}{{\sf pop}(\text{GW}(\vec{n},P,e_i))\geq k}.
\]
\end{lemma}
\begin{proof}
The proof proceeds by a coupling arguments.
\vspace{-2ex}
\paragraph{The upper bound.}
Consider a coupling  between the exploration process of the connected component of $v_i$ and $\text{GW}(\vec{n},P,e_i)$.
At each step, we will have queues $Q_{{\sf graph}}, Q_{\text{GW}}$ of unsaturated vertices in the graph
and in the Galton-Watson process, initially containing only $v_i$.

At each step, we take a vertex $w\in Q_{{\sf graph}}$, suppose its type is $i'$,
and an individual $\tilde{w}\in Q_{\text{GW}}$ of the same type $i'$. Letting $\vec{T_t} = (\vec{T_t}(1),\ldots,\vec{T_t}(r))$
be the number of vertices explored in the graph of each type so far, we sample $N_j\sim {\sf Binomial}(n_j - \vec{T_t}(j), p_{i',j})$
and $N_j(\text{GW}) \sim {\sf Binomial}(n_j, p_{i',j})$ in a correlated manner so that $N_j\leq N_j(\text{GW})$. We add $N_j$
vertices of type $j$ to $Q_{{\sf graph}}$ and $N_j(\text{GW})$ vertices of type $j$ to $Q_{\text{GW}}$.

Observe that at each step of the exploration process, for each type $j\in [r]$ the number of individuals of type $j$ in
$Q_{\text{GW}}$ is at least the number of individuals of type $j$ in $Q_{{\sf graph}}$. Also, at each step of the process
the number of vertices in the graph we have explored is at most the number of individuals explored in the Galton-Watson
process.
\vspace{-2ex}
\paragraph{The lower bound.}
For the lower bound, we use the same coupling, with one difference. If the total population explored in the process so far is at most $k$,
we take $w\in Q_{{\sf graph}}$, $\tilde{w}\in Q_{\text{GW}}$ of the same type $i'$, and sample
$N_j\sim {\sf Binomial}(n_j - \vec{T_t}(j), p_{i',j})$ and $N_j(\text{GW}) \sim {\sf Binomial}(n_j-\vec{T}(j), p_{i',j})$ in a correlated manner so that $N_j(\text{GW})\leq N_j$;
this is possible as $\vec{T_t}(j)\leq k \leq \vec{T}(j)$.
\end{proof}

As Lemma~\ref{GWSandwich} suggests, to understand the behaviour of connected components in random graphs we may need to
understand two different Galton-Watson processes, i.e. $\text{GW}(\vec{n}-\vec{T},P,e_i))$ and
$\text{GW}(\vec{n},P,e_i))$. Intuitively, these process could be thought as identical as long as $\max T_i = o(n)$ (at least
in the sense of extinction probabilities and rough statistics about the distribution of the population). For us, it will be enough to show that if the first of these process is super-critical, then so
is the second.

\begin{claim}\label{claim:Q_supercrit}
  For all $A,r\in\mathbb{N}$, $\eps>0$ there exists $N\in\mathbb{N}$ such that if $\vec{n}$ is $A$-balanced,
  $\Lambda$ is $\eps$-separated, $\rho(M)\geq 1+\eps$ and $n\geq N$, then the following holds for all non-zero
  initial configurations $x$. The process $\text{GW}(\vec{n} - n^{0.999}\vec{1}, P, x)$ is super-critical
  and furthermore $\rho(Q)\geq 1+\eps/2$ for the matrix $Q = ((n_j - n^{0.999})p_{i,j})_{i,j\in [r]}$.
\end{claim}
\begin{proof}
  Let $D_{M} = {\sf diag}(n_1,\ldots,n_r)$, $D_Q = {\sf diag}(n_1 - n^{0.999},\ldots,n_r - n^{0.999})$. Recall from the proof of Lemma~\ref{lemma:eigenbasis}
  that the top-most eigenvalue of $M$ (respectively $Q$) is the top-most eigenvalue of $D_M^{\half} P D_M^{\half}$ (respectively $D_Q^{\half} P D_Q^{\half}$).
  Since $D_M^{\half} P D_M^{\half}, D_Q^{\half} P D_Q^{\half}$ matrices are symmetric, we have by Weyl's inequality
  \begin{align*}
    \card{\rho(M) - \rho(Q)}^2
    =\card{\rho(D_M^{\half} P D_M^{\half}) - \rho(D_M^{\half} P D_M^{\half})}^2
    \leq \norm{D_M^{\half} P D_M^{\half} - D_M^{\half} Q D_M^{\half}}_2^2\\
    \qquad\qquad\qquad\qquad=\sum\limits_{i,j} (\sqrt{n_i} p_{i,j} \sqrt{n_j} - \sqrt{n_{i} - n^{0.999}} p_{i,j} \sqrt{n_j - n^{0.999}})^2.
  \end{align*}
  Note that $\sqrt{n_i - n^{0.999}}\geq \sqrt{n_i} - \sqrt{n^{0.999}}$ and similarly for $n_j$, so we get
  \[
    \card{\rho(M) - \rho(Q)}^2\leq
    \sum\limits_{i,j} (\sqrt{n^{0.999}} p_{i,j}(\sqrt{n_i}+\sqrt{n_j}))^2
    =n^{0.999}\sum\limits_{i,j}\frac{\lambda_{i,j}^2(\sqrt{n_i}+\sqrt{n_j})^2}{n_i n_j}.
  \]
  Note that for each $i$, $\frac{n}{r A}\leq n_i\leq n$, and $\lambda_{i,j}\leq 1/\eps$, so
  the last sum is at most $\frac{r^3 A 4}{n \eps^2}$, so
  $\card{\rho(M) - \rho(Q)}^2 \leq \frac{r^3 A 4}{\eps^2} \frac{n^{0.999}}{n}$, and as $n\geq N$,
  $\card{\rho(M) - \rho(Q)}^2\leq \frac{r^3 A 4}{\eps^2} N^{-0.0001}\leq \eps^2/4$ for large enough $N$.
  This proves the claim.
\end{proof}

\subsection{The distribution of the population size at intermediate steps}
Lastly, we need to gain some understanding to the distribution of the number of unsaturated individuals
after several steps in the parallel development of the Galton-Watson process. Let $\mathcal{T} = \text{GW}(\vec{n}, P)$
be some super-ciritcal Galton-Watson process, i.e. $\rho(M)\geq 1+\eps$.

For $i\in[r]$ and $t\in\mathbb{N}$, we denote by $U(t; \mathcal{T}; e_i)$ the (random) set of unsaturated vertices at time $t$
when we begin the exploration process of $\mathcal{T}$ from configuration $e_i$. We also denote by $S(t; \mathcal{T}; e_i)$ the (random) set of saturated vertices at time $t$.
We will often drop the process $\mathcal{T}$ from the notation if it is clear from context.

What can be said about the distribution of $\card{U(t; e_i)}$? In the standard (i.e., non multi-type) Galton-Watson process,
this question has a rather accurate answer, and this random variable is distributed roughly as a Poisson random variable (see~\cite{AlonSpencer} for example).
Here, intuitively the situation is similar, but as we are not aware of any literature addressing this question,
and instead establish cruder properties of this random
variable (that are enough for our application).

To gain some intuition, it is instructive to check that the expectations of these numbers grow exponentially with $t$. Indeed,
by the second item in Claim~\ref{claim:pump} for large enough $t$ we have $\norm{M^t e_i}_2\geq \delta \rho(M)^t$, and as
$\Expect{}{\card{U(t; e_i)}} = \norm{M^t e_i}_1\geq \norm{M^t e_i}_2$ that the expected growth is at least exponential. A similar
argument, using $\norm{M^t e_i}_1\leq \sqrt{r}\norm{M^t e_i}_2$, shows that the growth is at most exponential.

Unfortunately, expectation considerations do not seem to be enough to make our arguments go through, and we need to establish
stronger properties of $\card{U(t; e_i)}$. For example, intuition suggests that not only should the expectation of $\card{U(t; e_i)}$
be large, but that will be achieved in a very skewed way: either $\card{U(t; e_i)}$ will be $0$ with some probability, or else
it will be exponentially large in $t$. The following lemma proves that something along the lines indeed occurs.
\footnote{The relation between the number of steps in the process and the number of unsaturated individuals we get is likely to be suboptimal,
but this will essentially be irrelevant for us.}

\begin{lemma}\label{lem:better_than_avg}
  For all $A,r\in\mathbb{N}$, $\eps>0$ and $i\in [r]$
  there exists $\xi>0$, such that for all $H>0$ there is
  $t > 0$ such that
  \[
  \Prob{}{\card{U(t; e_i)}\geq H}\geq \xi.
  \]
\end{lemma}
\begin{proof}
  Let $T$ be large enough to be determined later. For each $t=1,\ldots,T$,
  we denote by $X_{t}$ the indicator random variable of that is $1$ if and only if
  $\card{U(t; e_i)}\leq H$.

  Let $\alpha_i$ be the survival probability of the Galton-Watson process starting
  at $e_i$. Note that
  \begin{equation}\label{eq2}
  \alpha_i \leq \prod\limits_{t = 1}^{T}\cProb{}{{\text{GW survives up to step }t}}{\text{GW survives up to step }t+1}.
  \end{equation}
  Observe that for each $t$,
  \[
  \cProb{}{{\text{GW survives up to step }t}}{\text{GW goes extinct in step }t+1}
  \geq \cProb{}{{{\text{GW survives up to step }t}}}{X_{t} = 1} c^{H},
  \]
  for some $c = c(A,r,\eps)>0$; this is because for each individual, the probability that they will have $0$ descendants is at least
  $c$, for an appropriately chosen $c$. Set $q_t = \cProb{}{{{\text{GW survives up to step }t}}}{X_{t} = 1}$; then
  moving to the complements events in~\eqref{eq2}, this bound yields
  \[
  \alpha_i \leq \prod\limits_{t = 1}^{T}(1-q_t c^{H})
  \leq e^{-c^{H}\sum\limits_{t=1}^{T} q_t},
  \]
  From Lemma~\ref{lem:GW_supcrit}, $\alpha_i\geq \delta(A,r,\eps)>0$, so we get that
  $e^{c^{H}\sum\limits_{t=1}^{T} q_t}\leq 1/\delta$, and hence $\sum\limits_{t=1}^{T} q_t\leq c^{-H}\log(1/\delta)$.
  Therefore, there is $t$ such that $q_t\leq \frac{c^{-H}\log(1/\delta)}{T}$.
  We thus take $T = 2c^{-H}\log(1/\delta)$ and get that $q_t \leq \half$, and
  the proof is concluded by noting that
  \[
   \Prob{}{\card{U(t; e_i)}\geq H}
   \geq \Prob{}{\text{GW survives to step }t} (1-q_t)
   \geq \alpha_i (1-q_t)
   \geq \frac{\delta}{2}.
  \]\qedhere
\end{proof}

Next, we prove that the total number of individuals that we explored in constantly many steps obeys a strong tail bound.
\begin{claim}\label{TailBound}
For all $t\in\mathbb{N}$, there exists $c>0$ such that for all $i\in[r]$
and $k\in\mathbb{N}$ we have
\[
\Prob{}{\card{U(t;e_i)\cup S(t;e_i)}\geq k}\leq e^{-c k}.
\]
\end{claim}
\begin{proof}
The proof is by induction on $t$.

\paragraph{Base case.}
When $t=1$, we have
\[
    \Prob{}{\card{U(t;e_i)\cup S(t;e_i)}\geq k}
    \leq
    \Prob{}{\sum\limits_{j=1}^r{\sf Binomial}(n_j, p_{i,j})\geq k},
\]
where all of the binomial random variables are independent. As $p_{i,j}\leq O_{A,r,\eps}(1/n)$,
the expectation of $\sum\limits_{j=1}^r{\sf Binomial}(n_j, p_{i,j})$ is $O_{A,r,\eps}(1/n)$, so
the claim follows from Chernoff's bound.

\paragraph{Inductive step.}
Suppose the claim is correct for $t$ and let $\xi = \xi(A,r,\eps)>0$ be a small constant to be determined later.
By the inductive hypothesis, $\Prob{}{\card{U(t;e_i)\cup S(t;e_i)}\geq \xi k}\leq e^{-\delta k}$. Condition
on $L = \card{U(t;e_i)\cup S(t;e_i)}$ and assume that $L\leq \xi k$; we have
\[
\Pr\Big[\card{U(t+1;e_i)\cup S(t+1;e_i)}\geq k~~\Big|~~\card{U(t;e_i)\cup S(t;e_i)} = L\Big]
\leq \Prob{}{\sum\limits_{j=1}^L B_j\geq (1-\xi)k}
\]
where each $B_j$ is a Binomial random variable of
the type ${\sf Binomial}(n_{i''}, p_{i',i''})$ for some $i',i''$,
and they are independent. The expectation of
$\sum\limits_{j=1}^L B_j$ is $O_{A,r,\eps}(L)\leq O_{A,r,\eps}(\xi k)\leq O_{A,r,\eps}(\xi k) \leq 0.01 k$
for suitable $\xi$ (this is how we choose $\xi$). Therefore, by Chernoff's bound
$\Prob{}{\sum\limits_{j=1}^L B_j\geq (1-\xi)k}\leq e^{-c k}$ for
some $c=C(A,r,\eps)>0$. Thus,
\begin{align*}
&\Prob{}{\card{U(t+1;e_i)\cup S(t+1;e_i)}\geq k}
\leq
\Prob{}{\card{U(t;e_i)\cup S(t;e_i)}\geq \xi k}\\
&\qquad+
\sum\limits_{L\leq \xi k}
\Prob{}{{\card{U(t;e_i)\cup S(t;e_i)} = L}}\cdot
\cProb{}{\card{U(t;e_i)\cup S(t;e_i)} = L}{\card{U(t+1;e_i)\cup S(t+1;e_i)}\geq k}\\
&\qquad\leq e^{-c' k}.\qedhere
\end{align*}
\end{proof}

\section{Proof of Theorem~\ref{thm:main}}\label{sec:proff_main}
In this section, we prove Theorem~\ref{thm:main}.

\subsection{The sub-critical case}\label{sec:sub_crit}
Let $i\in [r]$ and let $v_i\in V_i$ be a vertex of type $i$. By Lemma~\ref{GWSandwich}
we have
\[
\Prob{}{\card{{\sf CC}(v_i)}\geq \sqrt{n}}
\leq
\Prob{}{{\sf pop}(\text{GW}(\vec{n},P,e_i))\geq \sqrt{n}},
\]
and using the notations in the proof of Fact~\ref{fact:GW_subcript} we can bound this
by $\Expect{}{Z}/\sqrt{n}$, where $Z = \sum\limits_{t=1}^{\infty} X_t$. As in Fact~\ref{fact:GW_subcript},
$\Expect{}{Z}\leq \sum\limits_{t=1}^{\infty} \rho(M)^t\norm{e_i}_2\leq \frac{1}{\eps}$, so
$\Prob{}{\card{{\sf CC}(v_i)}\geq \sqrt{n}}$.

Let $E_v$ be the event that $\card{{\sf CC}(v)}\geq \sqrt{n}$. Then by linearity of expectation
$\Expect{}{\sum\limits_{v\in V}{1_{E_v}}}\leq \sqrt{n}$, so by Markov
\[
\Prob{}{\sum\limits_{v\in V}{1_{E_v}\geq n^{0.51}}}\leq 1/n^{0.01}.
\]
Thus, the probability that $\sum\limits_{v\in V}{1_{E_v}\leq n^{0.51}}$ is $1-o(1)$, and we argue that whenever this happens,
the largest connected component in $G$ has size at most $n^{0.51}$. Indeed, letting $S = \sett{v}{ 1_{E_v} = 1}$,
we have that each vertex outside $S$ has connected component of size at most $\sqrt{n}$; as for $v\in S$, note that
${\sf CC}(v)\subseteq S$, otherwise there is $u\in {\sf CC}(v)\setminus S$, and then as ${\sf CC}(v) = {\sf CC}(u)$,
the connected component of $v$ must be of size at most $\sqrt{n}$, in contradiction to $v\in S$. Thus,
$\card{{\sf CC}(v)}\leq \card{S} \leq n^{0.51}$.
\subsection{The super-critical case}
In this section, we begin the proof of the second item in Theorem~\ref{thm:main}.
\vspace{-2ex}
\paragraph{High level structure of the argument.}
Our proof has three steps.
\begin{enumerate}
  \item
  {\bf No middle ground.} We first show that for $k_{-} := \log^2(n)$, $k_{+} := n^{0.99}$, with probability
$1-o(1)$ the graph $G\sim G(\vec{n},P)$ does not have any connected components whose size
is between $k_{-}$ and $k_{+}$. This statement is the bulk of our proof.

  \item
  {\bf Sprinkling.} Second, we show that with probability $1-o(1)$, all components whose size exceeds $k_{+}$ will be connected,
thus in conjunction with the previous statement, we conclude that with high probability $G$ consists of components
of sizes $k_{-}$, and a single component that exceeds it.

  \item
  {\bf Estimating the size of the component.} Finally, by appealing to the Galton-Watson process again, we show that the probability of an vertex $v_i\in V_i$ to be
in the single component exceeding $k_{-}$ is $\alpha_i\pm o(1)$, where $\alpha_i$ is the probability the Galton-Watson
process with $P$ survives if it starts at initial configuration $e_i$. The proof is then concluded by
a simple application of Chebyshev's inequality.
\end{enumerate}

We remark that this proof outline is analogous to the proof in the simpler case of
the standard Erd\H{o}s-R\'{e}nyi model (see for example~\cite{AlonSpencer}). The proofs of
the first and third steps however require more work.

\subsubsection*{No middle ground: proof overview}
Next, we elaborate on the main step in the argument outlined above, in which we show that
with probability $1-o(1)$, $G$ does not contain components of size between $k_{-}$ and $k_{+}$.

Suppose we wish to explore the connected component of a given vertex $v_i\in V_i$ for some type $i\in [r]$.
Towards this end, it is useful to consider a coupling between what's happening in the graph $G$, and in a corresponding execution of the
Galton-Watson process with $P$ on initial configuration $v_i$.

The only difference between the two procedures is that in the graph case there is no
replacement, hence after we have explored $T$ vertices spread according to the vector $\vec{T} = (T_1,\ldots,T_r)\in\mathbb{N}^r$,
the neighbours of type distribution of a new unsaturated vertex is slightly different. Namely, for each $j\in [r]$, the number of neighbours of type $j$ for a vertex $v_i$ of type $i$ is distributed as ${\sf Binomial}(n_j - T_j, p_{i,j})$, whereas in Galton-Watson process, the number of offsprings of type $j$ for an element of type $i$ would still be $\sim{\sf Binomial}(n_j,p_{i,j})$. However, as these random variables are anyway ``close'' to each other (at least as long as $T$ is not too large), this difference should be thought of as minor. For clarity, we ignore this issue in this overview.

The main issue in working with the original Galton-Watson process, is that it is not ``immediately clear'' whether the process is super critical or
not in a very local point of view. For example, it could be the case that (a) for each $j\in [r-1]$, an individual of type $j$ can only have descendants
of type $j+1$, and the average number of them could be small -- say $\half$, and (b) for $j = r$, an individual of type $j$ is likely to have many descendants
of type $j$ -- say at least $2$ in average. This process is connected and super-critical, however if we start with an individual of type $j=1$, it will
be hard for us to identify it locally. Namely, we would need to make at least $r$ steps in the process in order to have a chance of having an offspring
of type $r$, in which case it is ``clear'' that the process has a chance to never go extinct.

To circumvent this issue, we consider a modified Galton-Watson process, call it $\text{GW}_{{\sf modified}}$,
in which a step corresponds to a batch of constantly many parallel steps in the original Galton-Watson process.
We will also explore the connected component of $v$ in a process that is analogous to the modified Galton-Watson process.
Using this idea, we are able to show that:
\begin{enumerate}
  \item in order for the connected component to exceed $k_{-}$ in size, except for negligible probability,
  the modified Galton-Watson process must have gone on for at least $\Omega(k_{-})$ steps.
  \item For $T = O(k_{-})$, if graph exploration process survives at least $T$ steps, then except for negligible probability it
  has at least $\Omega(T)$ unsaturated vertices at step $T$.
  \item If the graph exploration process at least $U$ unsaturated vertices at some point, then except for
  probability $e^{-\Omega(U)}$, the connected component of $v_i$ will exceed $k_{+}$ in size. Intuitively,
  this is true since each one of the unsaturated vertices can be thought of as initiating a Galton-Watson process, which has
  positive probability of not going extinct. Furthermore, the events of ``not going extinct'' are not too negatively correlated.
\end{enumerate}

Below is a formal statement of a lemma that easily implies the ``no middle ground'' step.
\begin{lemma}\label{lem:main_no_mid}
  For all $A, r\in\mathbb{N}$, $\eps>0$, there exists $N>0$, $\delta>0$ such that if $\vec{n}$ is $A$-balanced, $n\geq N$, $\rho(M)\geq 1+\eps$
  and $\Lambda$ is $\eps$-separated, then for each $k_{-}\leq k\leq k_{+}$ and $v\in V$,
  \[
  \Prob{}{\card{{\sf CC}(v)} = k}\leq 3e^{-\delta k_{-}}.
  \]
\end{lemma}

\subsection{Proof of Lemma~\ref{lem:main_no_mid}}
\subsubsection{The modified Galton-Watson process}
Consider the Galton-Watson process $\mathcal{T} = \text{GW}(\vec{n} - k_{+}\vec{1}, P)$. Note that by Claim~\ref{claim:Q_supercrit}
this process is supercritical, and for its corresponding matrix $Q = ((n_j - k_+)p_{i,j})_{i,j\in [r]}$ has top-most eigenvalue at least
$1+\eps/2$. We may therefore apply Lemma~\ref{lem:better_than_avg}: first, we pick $\xi_1,\ldots,\xi_r$ from the Lemma~\ref{lem:better_than_avg} for $i=1,\ldots,r$,
set $\xi = \min(\xi_1,\ldots,\xi_r)$ and then take $H = 11\xi^{-r}$. Then, we take $t_1,\ldots,t_r$ from Lemma~\ref{lem:better_than_avg} for our choice of $H$.

We define a modified Galton-Watson process $\mathcal{T}_{{\sf modified}}$ as follows.
The process starts at some vertex $v\in V$, say of type $i$, and maintains as before lists ${\sf LU}(t)$ and
${\sf LS}(t)$ of unsaturated and saturated vertices at step $t$. Upon exploring
a vertex $w\in {\sf LU}(t)$, if its type is $j$, we run the original Galton-Watson process on $j$ for $t_j$ steps (in the ``parallel view''),
and then update ${\sf LS}(t)$ and ${\sf LU}(t)$ accordingly.

\subsubsection{The sandwiching random variables}
We will want to estimate the number of explored vertices as well as the number of unsaturated vertices at various times in the process.
However, as the random variables counting the number of newly explored vertices at each step are dependent, we are not able to use strong
concentration bounds such as Chernoff's inequality. To overcome this difficulty, we will use auxiliary random variable to lower bound the
number of newly found unsaturated vertices, and upper bound the total number of explored vertices, at each point in the process.

For each $i$, let $R_i$ be the random variable counting the number of unsaturated vertices found in a single step of the modified Galton-Watson $\mathcal{T}_{\sf modified}$
process in a single step from configuration $e_i$. Let $C>0$ be a large enough constant to be determined later, and set
$R = \min(R_1,\ldots,R_r, C)$ (where $R_1$,\ldots $R_r$ are independent). It is clear that $R$ is stochastically dominated by
each $R_i$, and we next show that it obeys a tail bound and its the expectation of $R$ is strictly larger than $1$.
\begin{claim}\label{claim:fake_lb_expect}
  There exists $C = C(A,r,\eps)>0$ such that in the above set-up, $\Expect{}{R}\geq 10$.
\end{claim}
\begin{proof}
  First, consider the random variable $\tilde{R} = \min(R_1,\ldots,R_r)$. Then
  \[
  \Expect{}{\tilde{R}} = \sum\limits_{h=1}^{\infty}\Prob{}{\tilde{R}\geq h}
  \geq \sum\limits_{h=1}^{H}\Prob{}{\tilde{R}\geq h}
  = \sum\limits_{h=1}^{H}\prod\limits_{i=1}^{r}\Prob{}{R_i\geq h}
  \geq \sum\limits_{h=1}^{H} \xi^r = H\xi^r = 11.
  \]
  Next, write
  \[
  \Expect{}{R} = \Expect{}{\tilde{R}} - \Expect{}{\tilde{R} 1_{R > C}}
  \geq \Expect{}{\tilde{R}} - C\sum\limits_{h=C+1}^{\infty}\Prob{}{\tilde{R}\geq h}
  \geq 100 - C\sum\limits_{h=C+1}^{\infty}\Prob{}{\tilde{R}\geq h}.
  \]
  To upper bound the last sum, we use Claim~\ref{TailBound} to bound $\Prob{}{\tilde{R}\geq h}\leq e^{-\beta h}$ for
  some $\beta = \beta(t) = \beta(A,r,\eps) > 0$. Therefore,
  \[
  C\sum\limits_{h=C+1}^{\infty}\Prob{}{\tilde{R}\geq h}
  \leq \sum\limits_{h=C+1}^{\infty}e^{-\beta h}
  \leq \frac{C e^{-\beta C}}{1-e^{-\beta}},
  \]
  and taking $C$ large enough, this expression is at most $1$.
\end{proof}

We also construct a random variable $W$ that will serve as an upper bound for the population.
Let $\mathcal{T}^2 = \text{GW}(\vec{n}, P)$, and consider $\mathcal{T}_{{\sf modified}}^2$ to
be the modified Galton-Watson process of $\mathcal{T}^2$. Namely, this process maintains lists of
saturated and unsaturated vertices, and at each step picks an unsaturated vertex -- say $w$ whose
type is $j$ -- and then performs $t_j$ parallel steps of $\mathcal{T}^2$ on $w$. Finally, the lists
of saturated and unsaturated individuals are updated appropriately.

Let $\tilde{R}_i$ be the random variable measuring the total population explored in a single step of
$\mathcal{T}^2_{{\sf modified}}$ from configuration $e_i$, and define $W = \max(\tilde{R}_1,\ldots,\tilde{R}_r)$.
\begin{claim}\label{claim:fake_up_tail}
  There exists $\delta = \delta(A,r,\eps)>0$ such that for all $h$,
  $\Prob{}{W\geq h}\leq r e^{-\delta h}$.
\end{claim}
\begin{proof}
  Immediate from Claim~\ref{TailBound} and the union bound.
\end{proof}

\subsubsection{The main argument: the proof of Lemma~\ref{lem:main_no_mid}}
Fix $v\in V$. We explore the connected component of $v$ using a process $\mathcal{G}$ as follows. Throughout the process,
we maintain lists ${\sf LU}(t)$ and ${\sf LS}(t)$ of unsaturated and saturated vertices at step $t$ respectively.
At each step, we take an arbitrary $w\in {\sf LU}(t)$ -- say it's type is $i$, and then explore the neighbourhood
of radius $t_i$ around $w$ (analogously to the modified Galton-Watson process). We add $w$ and any other
new saturated vertex found in this step to ${\sf LS}(t)$, and add all new unsaturated vertices to ${\sf LU}(t)$.

Let $D$ be a random variable denoting the total number of steps in the process $\mathcal{G}$.
Let $U(\ell)$ be the number of unsaturated vertices found in the $\ell$-th step of $\mathcal{G}$;
note that the number of unsaturated vertices in step $L$ is then $\sum\limits_{\ell=1}^{L} U(\ell) - L$.
We also denote by $S(\ell)$ the number of saturated
vertices introduced in the $\ell$-th step of the modified process.

\paragraph{The coupling.} We now describe a coupling procedure that will help us analyze the above random variables.
Suppose that at the $\ell$-th step of the process. If we explored more than $k_+$ vertices we halt.
Otherwise, we take an unsaturated vertex $w$ -- say of type $j$, and perform a step of $\mathcal{G}$.
\begin{enumerate}
  \item Sample a random variable $R(\ell)$ which is an independent copy of $R$ and $W(\ell)$ independent copy of $W$
  %\item Sample a random variable $\tilde{S}(\ell)$ which is the total number of saturated vertices in $G(\vec{n}, Q)$ when we do a step from a vertex of type $j$.
%  \item Sample a random variable $\tilde{U}(\ell)$ which is an independent copy of the total number of unsaturated vertices in $G(\vec{n} - \vec{k_{+}}, Q)$ when we do a step from a vertex of type $j$.
  \item Sample $U(\ell), S(\ell)$ conditioned on $U(\ell)\geq R(\ell)$ and $S(\ell)\leq W(\ell)$ (we justify below why this is possible).
\end{enumerate}
To see that the last step is possible, let $\tilde{U}(\ell)$ be the total number of unsaturated vertices in $G(\vec{n} - \vec{k_{+}}, P)$
when we do a step from a vertex of type $j$. Note that $U(\ell)$ stochastically dominates $\tilde{U}(\ell)$, which in return stochastically dominates $R(\ell)$.
Since the stochastic domination relation is transitive, by Fact~\ref{fact:domination_coupling} we can do the sampling of $U(\ell)$ as above.
The argument for $S(\ell)$ is analogous.

This process halts if the total population exceeds $k_+$, or we ran out of unsaturated vertices. To make the analysis simpler,
it will be helpful for us to imagine we still sample copies of $R$, namely $R(\ell)$'s even after the above process terminates,
so we in fact have the random variables $R(\ell)$ for all $\ell\in\mathbb{N}$. We also remark that they are independent.

\paragraph{Analysis.}
Let $\zeta$ be a parameter to be chosen later. To bound the probability in question, we first use the union bound to say
\begin{equation}\label{eq3}
  \Prob{}{\card{{\sf CC}(v)} = k} \leq
  \underbrace{\Prob{}{D\leq \zeta k_{-}\text{ and }\card{{\sf CC}(v)} \geq k_{-}}}_{(\rom{1})}
  +
  \underbrace{\Prob{}{D\geq \zeta k_{-}\text{ and }\card{{\sf CC}(v)} \leq k_{+}}}_{(\rom{2})}.
\end{equation}

We first show that it is very unlikely that the process terminates early and explores many vertices.
\begin{claim}\label{claim:upper_bound_total_pop}
  There exist $\zeta = \zeta(A,r,\eps)>0$, $\delta_1 = \delta_1(A,r,\eps)>0$ such that
  $(\rom{1})\leq e^{-\delta_1 k_{-}}$.
%  \[
%  \Prob{}{D\leq \zeta k_{-} \text{ and } \sum\limits_{\ell=1}^{D} \tilde{S}(\ell)\geq k_{-}}\leq e^{-\delta k_{-}}
%  \]
\end{claim}
\begin{proof}
First, since by the coupling we have that $\card{{\sf CC}(v)}\leq \sum\limits_{\ell=1}^{D} W(\ell)$ if $\card{{\sf CC}(v)}\leq k_{+}$ and
otherwise $\sum\limits_{\ell=1}^{D} W(\ell)\geq k_{+}$, it follows that
\[
    (\rom{1})\leq \Prob{}{D\leq \zeta k_{-} \text{ and } \sum\limits_{\ell=1}^{D} W(\ell)\geq k_{-}}.
\]
By Claim~\ref{claim:fake_up_tail}, each $W(\ell)$ is a sub-Gaussian random variable with constant $\delta(A,r,\eps)>0$, and
in particular its expectation bounded by a constant $K(A,r,\eps)>0$. Thus, by Chernoff's bound
\[
\Prob{}{\sum\limits_{\ell=1}^{d} W(\ell)\geq 2K(A,r,\eps) d}\leq e^{-\delta' d},
\]
for some $\delta' = \delta'(A,r,\eps) > 0$. Therefore, as the probability in question is at most
$\Prob{}{\sum\limits_{\ell=1}^{\xi k_{-}} W(\ell)\geq k_{-}}$, it is upper bounded by $\leq e^{-\delta'' k_{-}}$.
\end{proof}
We fix $\zeta$ from Claim~\ref{claim:upper_bound_total_pop} henceforth.
\begin{claim}\label{claim:upper_bound_total_pop_2}
  There exists $\delta_2 = \delta_2(A,r,\eps)>0$ such that
  $(\rom{2}) \leq 2e^{-\delta_2 k_{-}}$.
\end{claim}
\begin{proof}
  We first argue that $\Prob{}{\sum\limits_{\ell=1}^{\zeta k_{-}} R(\ell)\leq 5\zeta k_{-}}\leq e^{-\delta' k_{-}}$ for some
  $\delta' = \delta'(A,r,\eps)>0$. Indeed, the random variables $R(\ell)$ are independent, bounded between $0$ and $C = C(A,r,\eps)>0$
  and have expectation at least $10$ by Claim~\ref{claim:fake_lb_expect}, so this just follows from Chernoff's bound. Therefore,
  we can write
  \[
    (\rom{2}) \leq e^{-\delta' k_{-}} + \Prob{}{D\geq \zeta k_{-}\text{ and } \sum\limits_{\ell=1}^{\zeta k_{-}} R(\ell)\geq 5\xi k_- \text{ and }\card{{\sf CC}(v)} \leq k_{+}}.
  \]
  The rest of the proof is devoted to bounding the probability on the right hand side.
  Note that whenever the event on the right hand side holds, we have that the number of unsaturated vertices explored by time $\zeta k_{-}$ is
  $\sum\limits_{\ell=1}^{\zeta k_{-}} R(\ell) - \zeta k_{-}$, and in particular at least $4\zeta k_-$, so we may bound this probability by
  \[
    \Prob{}{\card{{\sf LU}(\zeta k_{-})} \geq 4\zeta k_- \text{ and }\card{{\sf CC}(v)} \leq k_{+}}
    \leq
    \cProb{}{\card{{\sf LU}(\zeta k_{-})} \geq 4\zeta k_-}{\card{{\sf CC}(v)} \leq k_{+}}
  \]
  Consider the process at time $\zeta k_-$, and condition on the set of vertices in ${\sf LU}(\zeta k_-)$ and
  the set of vertices ${\sf LS}(\zeta k_-)$, and let $s_1,s_2$ be sizes of these sets respectively. The contribution of the case
  $s_1+s_2 > k_{+}$ to the probability in question is $0$, so we assume henceforth that $s_1+s_2\leq k_{+}$.
  Fix an ordering $w_1,\ldots,w_{\zeta k_{-}}$ on some $\zeta k_-$ vertices from ${\sf LU}(\zeta k_-)$. For each $1\leq \ell\leq \zeta k_{-}$, let $E_{\ell}$ be
  the probability that starting the exploration process from $w_{\ell}$ we find at least $k_+$ vertices
  (without using vertices from ${\sf LU}(\zeta k_{-})\cup {\sf LS}(\zeta k_{-})$). Let $q_\ell = \cProb{}{\bigcap_{\ell'<\ell} \bar{E_{\ell'}}}{\bar{E_{\ell}}}$.

  We bound each $q_{\ell}$ from above. Conditioning on $\bigcap_{\ell'<\ell} \bar{E_{\ell'}}$ and further on the vertices explored in the connected components
  of $w_1,\ldots,w_{\ell-1}$. Let $\vec{T}$ be the type statistics of these vertices and of the vertices in ${\sf LU}(\zeta k_{-})\cup {\sf LS}(\zeta k_{-})$.
  We note that the event $E_{\ell}$ then is the event that the vertex $w_{\ell}$ has a connected component of size at most $k_{+}$ in
  $G(\vec{n} - \vec{T}, P)$. Therefore, by Lemma~\ref{GWSandwich}, letting $j$ be the type of $w_{\ell}$ we have
  \[
    q_\ell \leq
    1 - \Prob{}{{\sf pop}(\text{GW}(\vec{n}-\vec{T}-k_{+} \vec{1}, P, e_j))\geq k_{+}}
    \leq
    1 - \Prob{}{\text{GW}(\vec{n}-\vec{T}-k_{+} \vec{1}, P, e_j))\text{ survives}}.
  \]
  Note that by Claim~\ref{claim:Q_supercrit}, the process $\text{GW}(\vec{n}-\vec{T}-k_{+} \vec{1}, P, e_j))$ is supercritical
  and its matrix $Q$ has $\rho(Q)\geq 1+\eps/2$. Thus, by Lemma~\ref{lem:GW_supcrit} the survival probability of this process
  is at least some $\delta'' = \delta''(A,r,\eps)>0$, and so $q_w\leq 1-\delta''$. Thus, we get that
  \[
  \cProb{}{\card{{\sf LU}(\zeta k_{-})} \geq 4\zeta k_-}{\card{{\sf CC}(v)} \leq k_{+}}
  =\prod\limits_{\ell=1}^{\zeta k_{-}} q_{\ell}\leq (1-\delta'')^{\zeta k_{-}}\leq
  e^{-\delta'' \zeta k_{-}}.
  \]
  Therefore, the claim is proved for $\delta_2 = \min(\delta', \delta''\zeta)$.
\end{proof}
Plugging Claims~\ref{claim:upper_bound_total_pop},~\ref{claim:upper_bound_total_pop_2} into~\eqref{eq3} we get that
$\Prob{}{\card{{\sf CC}(v)} = k}\leq 3 e^{-\delta_3 k_{-}}$ for $\delta_3 = \min(\delta_1,\delta_2)$, finishing the proof
of Lemma~\ref{lem:main_no_mid}.\qed

\subsection{Analysis of the size of the giant component}
\subsubsection{Sprinkling: the uniqueness of the giant component}
\begin{claim}\label{claim:sprinkle_aux}
  For all $A,r\in\mathbb{N}$, $\eps>0$ there exists $\delta>0$ such that the following holds.
  Suppose $\vec{n}$ is $A$-balanced and $\Lambda$ is connected and $\eps$-separated, and let
  $i\in [r]$ and $j\in [r]$ be such that $\lambda_{i,j} > \eps$.
  Then for all $k$ we have
  \[
  \Prob{G\sim G(\vec{n}, P)}{\exists S\subseteq V_i, \card{S} = k\text{ such that } \card{N(S)\cap V_j}\leq \frac{\delta k}{\log n}}
  \leq e^{-\delta k}.
  \]
\end{claim}
\begin{proof}
  Let $\delta = 1/(4A)$. Denote $k' = \lceil \frac{\delta k}{\log n}\rceil$. By the union bound,
  \begin{align*}
    &\Prob{G\sim G(\vec{n}, P)}{\exists S\subseteq V_i, \card{S} = k\text{ such that } \card{N(S)\cap V_j}\leq k'}\\
    &\qquad\leq \sum\limits_{T\subseteq V_j, \card{T} = k'} \Prob{G\sim G(\vec{n}, P)}{\exists S\subseteq V_i, \card{S} = k\text{ such that } N(S)\cap V_j\subseteq T)}\\
    &\qquad\leq \sum\limits_{T\subseteq V_j, \card{T} = k'}\prod\limits_{\substack{s\in S\\ t\not\in T}}\left(1-\frac{\eps}{n}\right)\\
    &\qquad= \binom{n_j}{k'}\left(1-\frac{\eps}{n}\right)^{k\cdot(n_j-k')}\\
    &\qquad\leq \binom{n_j}{k'}e^{-\frac{\eps}{n}k\cdot(n_j-k')},
  \end{align*}
  where in the last inequality we used $e^{-x}\geq 1-x$.
  Note that $\binom{n_j}{k'}\leq n^{k'}\leq e^{k' \log n}\leq e^{\delta k}$, and
  $n_j\geq n/A$, $k'\leq n/2A$, so we get that the above probability is upper bounded by
  $e^{-\delta k} e^{-k/2A}\leq e^{-k/4A}$.
\end{proof}

Define the matrix $P' = \left(1_{\lambda_{i,j} > 0} \frac{1}{n^{1.95}}\right)_{i,j\in[r]}$.
\begin{claim}\label{claim:edge_between_sets}
  Let $i,j\in [r]$ be such that $\lambda_{i,j}\geq \eps$, and let $S\subseteq V_i$, $T\subseteq V_j$ each
  be of size at least $n^{0.98}$. Then
  \[
    \Prob{G'\sim G(\vec{n}, P')}{N(S)\cap T=\emptyset}\leq e^{-n^{0.01}}.
  \]
\end{claim}
\begin{proof}
  The probability in question is
  \[
  \left(1-\frac{1}{n^{1.95}}\right)^{\card{S}\cdot \card{T}}\leq
  e^{-\frac{\card{S}\card{T}}{n^{1.95}}}
  \leq e^{-n^{0.01}}.\qedhere
  \]
\end{proof}

\begin{claim}\label{claim:bd_sd}
  Let $G'\sim G(\vec{n},P')$, $G\sim G(\vec{n},P)$. Then
  ${\sf SD}(G\cup G', G)\leq o(1)$.
\end{claim}
\begin{proof}
  Let $X$ be the set of potential edges, i.e. $e = (v_i,v_j)$ where $v_i\in V_i$, $v_j\in V_j$
  for  types $i,j\in [r]$ such that $\lambda_{i,j}\geq \eps$. For each $e\in X$,
  let $Z_e$ be the indicator random variable that $e\in G$, and let $Z'_e$ be the indicator random variable
  that $e\in G'$. Then the statistical distance between $G$ and $G\cup G'$ is at most
  the statistical distance between the ensemble of random variables $(Z_e)_{e\in E}$ and
  $(Z_e\lor Z'_e)_{e\in E}$, and each one of these ensembles consists of independent random variables.
  We compute the KL-divergence between the two ensembles. Due to independence, this KL-divergence is equal to
  \[
  \sum\limits_{e\in E} {\sf D}_{{\sf KL}}(Z_e; Z_e\lor Z'_e)
  =\sum\limits_{e\in E, e=(v_i,v_j)} {\sf D}_{{\sf KL}}({\sf Bernouli}(p_{i,j}); {\sf Bernouli}(p'_{i,j})),
  \]
  where $p'_{i,j} = p_{i,j} + (1-p_{i,j})\frac{1}{n^{1.95}}$. A standard computation using $\log(z)\leq \frac{z-1}{\ln 2}$
  shows that
  \[
  {\sf D}_{{\sf KL}}({\sf Bernouli}(p_{i,j}); {\sf Bernouli}(p'_{i,j}))\leq
  \frac{(p_{i,j} - p'_{i,j})^2}{p'_{i,j}(1-p'_{i,j})\ln 2}
  \leq \frac{n^{-3.9}}{\eps n^{-1} 2^{-1}\ln 2}
  \leq \frac{2}{\eps \ln 2} n^{-2.9},
  \]
  and plugging this about we get that
  \[
  {\sf D}_{{\sf KL}}((Z_e)_{e\in E}; (Z_e\lor Z'_e)_{e\in E})
  =\sum\limits_{e\in E} {\sf D}_{{\sf KL}}(Z_e; Z_e\lor Z'_e)
  \leq n^2 \frac{2}{\eps \ln 2} n^{-2.9} = o(1).
  \]
  Now Pinsker's inequality gives
  ${\sf SD}((Z_e)_{e\in E}; (Z_e\lor Z'_e)_{e\in E})\leq \sqrt{{\sf D}_{{\sf KL}}((Z_e)_{e\in E}; (Z_e\lor Z'_e)_{e\in E})/2} = o(1)$.
\end{proof}

\begin{lemma}\label{lem:sprinkle}
  $\Prob{G\sim G(\vec{n}, P)}{\exists u, v\text{ such that } {\sf CC}(u)\neq {\sf CC}(v)\text{ each of size at least }k_{+}}
  \leq o(1)$.
\end{lemma}
\begin{proof}
  Denote the probability in question by $q$.

  We first argue that except for probability $e^{-\sqrt{n}}$, each connected component of size $k_{+}$ contains
  least $n^{0.98}$ vertices of each type. More specifically, let $E$ be the event whose complement is
  \[
   \bar{E} = \set{\exists k\geq n^{0.98},\exists S\subseteq V_i, \card{S} = k\text{ such that } \card{N(S)\cap V_j}\leq \frac{\delta k}{\log n}}.
  \]
  By Claim~\ref{claim:sprinkle_aux} and the union bound, the probability of $\bar{E}$ is at most $n\cdot e^{-\delta n^{0.98}}\leq e^{-\sqrt{n}}$,
  where the last inequality holds for large enough $n$. We now argue that if $E$ holds, then any connected component of size at least $n^{0.99}$
  contains at least $n^{0.98}$ vertices of each type. First, note that there is a type $i$ such that
  $\card{C\cap V_i}\geq n^{0.99}/r$. Since $E$ holds, for any $j$ such that $\lambda_{i,j} > 0$ we have that
  $\card{C\cap V_j}\geq \frac{\delta n^{0.99}}{\log n}$, and as $\Lambda$ is connected we get that for all $j$ it holds that
  $\card{C\cap V_j}\geq \frac{\delta^r n^{0.99}}{\log^r n}\geq n^{0.98}$.

  Sample $G\sim G(\vec{n},P)$ and $G'\sim(\vec{n}, P')$. Then
  \begin{align*}
  q\leq
  &{\sf SD}(G, G\cup G') + \Prob{G, G'}{\bar{E}}\\
  &+ \cProb{G,G'}{G\in E}{\text{number of connected components in }G\cup G'\text{ whose size is at least }k_{+}\text{ is at least }2}.
  \end{align*}
  By Claim~\ref{claim:bd_sd}, the statistical distance between $G$ and $G\cup G'$ is $o(1)$. By the argument above,
  the probability of $\bar{E}$ is $o(1)$. Finally, for the last probability, condition on $G$ and let $C_1,C_2,\ldots,C_{\ell}$
  be all connected components of $G$ of size at least $k_{+}$. Then By Claim~\ref{claim:edge_between_sets}, for any distinct $1\leq \ell',\ell''\leq \ell$,
  the probability there is no edge between $C_{\ell'}$ and $C_{\ell''}$ is at most $e^{-n^{0.01}}$, so by the union bound the last probability
  is at most $\ell^2 e^{-n^{0.01}}\leq n^2 e^{-n^{0.01}} = o(1)$. Combining everything, we get that $q=o(1)$.
\end{proof}
\subsubsection{The size of the giant component}
Let $E$ be the event that $G\sim G(\vec{n},P)$ contains only components of size at most $k_{-}$, and
at most a single connected component whose size exceeds $k_{+}$. Using Lemma~\ref{lem:main_no_mid} we get that the probability that
$G\sim G(\vec{n},P)$ contains a connected component
whose size is between $k_{-}$ and $k_{+}$ is at most $n\cdot e^{-\delta k_{-}} = o(1)$ for large enough $n\geq N$. Secondly, by Lemma~\ref{lem:sprinkle}
the probability $G$ contains more than one connected component of size $\geq k_{+}$ is at most $o(1)$. Therefore, $\Prob{}{E} = 1-o(1)$.

For the rest of the proof, we sample $G\sim G(\vec{n},P)$ conditioned on $E$.
Let $W$ be the random variable which is the connected component of size at least $k_{+}$. For each $v$, let $Z_v$ be the indicator random variable
of $v\not\in W$, i.e. that $\card{{\sf CC}(v)}\leq k_{-}$.

Let $\alpha_i,$ be the survival probability of $\text{GW}(\vec{n}, P, e_i)$.

\begin{claim}\label{claim:expectation_compute1}
  For all $i\in [r]$ and $v\in V_i$ it holds that $\Expect{}{Z_v} \leq 1-\alpha_i + o(1)$.
\end{claim}
\begin{proof}
  We have
  \[
  \cExpect{G\sim G(\vec{n},P)}{E}{Z_v}
  =\cProb{G\sim G(\vec{n},P)}{E}{\card{{\sf CC}(v)}\leq k_{-}}
  \leq (1+o(1))\Prob{G\sim G(\vec{n},P)}{\card{{\sf CC}(v)}\leq k_{-}},
  \]
  where we used the fact that $\Prob{}{E}\geq 1-o(1)$. By Lemma~\ref{GWSandwich},
  \[
  \Prob{G\sim G(\vec{n},P)}{\card{{\sf CC}(v)}\leq k_{-}}
  \leq
  \Prob{}{{\sf pop}(\text{GW}(\vec{n}-k_{-}\vec{1},P,e_i))\leq k_{-}}.
  \]
  By a coupling argument, the last probability is at most
  $\Prob{}{{\sf pop}(\text{GW}(\vec{n},P,e_i))\leq k_{-}} + o(1)$, which is at
  most $1-\alpha_i + o(1)$.
\end{proof}

\begin{claim}\label{claim:expectation_compute2}
  For all $i\in [r]$ and $v\in V_i$ it holds that $\Expect{}{Z_v} \geq 1-\alpha_i + o(1)$.
\end{claim}
\begin{proof}
  As before,
  \[
  \cExpect{G\sim G(\vec{n},P)}{E}{Z_v}
  =\cProb{G\sim G(\vec{n},P)}{E}{\card{{\sf CC}(v)}\geq k_{-}}
  \geq (1-o(1))\Prob{G\sim G(\vec{n},P)}{\card{{\sf CC}(v)}\geq k_{-}},
  \]
  where we used the fact that $\Prob{}{E}\geq 1-o(1)$. By Lemma~\ref{GWSandwich},
  \[
  \Prob{G\sim G(\vec{n},P)}{\card{{\sf CC}(v)}\geq k_{-}}
  \geq
  \Prob{}{{\sf pop}(\text{GW}(\vec{n}-k_{-}\vec{1},P,e_i))\geq k_{-}}.
  \]
  By a coupling argument, the last probability is at most
  $\Prob{}{{\sf pop}(\text{GW}(\vec{n},P,e_i))\geq k_{-}} - o(1)$,
  and we claim this is at least $1-\alpha_i + o(1)$. To see that, write
  \begin{align*}
  \Prob{}{{\sf pop}(\text{GW}(\vec{n},P,e_i))\geq k_{-}}
  &\geq
  \Prob{}{\text{GW}(\vec{n},P,e_i)\text{ survives}}\\
  &-
  \Prob{}{\text{GW}(\vec{n},P,e_i)\text{ goes extinct and }{\sf pop}(\text{GW}(\vec{n},P,e_i))\geq k_{-}}.
  \end{align*}
  The first probability is $1-\alpha_i$, and to finish the proof we show that the second probability is $o(1)$.
  To see that, consider the parallel view on the Galton-Watson
  process, and let $D$ be a random variable indicating the number of steps it takes until halting, and take $T = \sqrt{\log(k_{-})}$. Write
  \begin{align*}
    &\Prob{}{\text{GW}(\vec{n},P,e_i)\text{ goes extinct and }{\sf pop}(\text{GW}(\vec{n},P,e_i))\geq k_{-}}\\
    &\leq
    \underbrace{\Prob{}{D\leq T\text{ and }{\sf pop}(\text{GW}(\vec{n},P,e_i))\geq k_{-}}}_{(\rom{1})}
    +
    \underbrace{\Prob{}{D\geq T\text{ and } \text{GW}(\vec{n},P,e_i)\text{ goes extinct and }{\sf pop}}}_{(\rom{2})}.
  \end{align*}

  For $(\rom{1})$, note that in $T$ steps, the expected population
  size of the Galton-Watson process is at most $\sum\limits_{t=1}^{T}\rho(M)^{t} \leq \frac{\rho(M)^{T+1}}{\eps}$, which by inequality~\eqref{eq:rho_upper_bd}
  (recall that $\rho(M) = \rho(\Lambda)$) is at most $r^{T+1}/\eps^{T+2}$. Thus, by Markov
  \[
    (\rom{1})
    =
    \Prob{}{\text{population of }\text{GW}(\vec{n},P,e_i)\text{ is at least }k_{-}\text{ in T steps}}
    \leq \frac{r^{T+1}/\eps^{T+2}}{k_{-}}
    =o(1).
  \]

  For $(\rom{2})$, a verbatim repeat of the argument in Claim~\ref{claim:upper_bound_total_pop_2} shows that $(\rom{2}) = o(1)$ (we omit the repetition),
  and we are done.
\end{proof}

\begin{claim}\label{claim:bd_cov}
  Let $i,j\in [r]$ be types (not necessarily distinct), and let $v\in V_i$, $u\in V_j$ be two distinct vertices. Then
  $\Expect{}{Z_v Z_u}\leq \Expect{}{Z_v}\Expect{}{Z_u} + \frac{2k_{-}^2}{n_j}$.
\end{claim}
\begin{proof}
  By definition,
  \[
  \Expect{}{Z_v Z_u} =
  \Prob{}{Z_v = 1}\cProb{}{Z_v=1}{Z_u=1}
  \leq \Expect{}{Z_v}\cProb{}{Z_v=1}{Z_u=1}.
  \]
  To analyze the last probability, condition on the connected component of $v$, call it $S$.
  By symmetry, for each $u'\in V_j$ the probability that $u'\in S$ is the same, and as $\card{S}\leq k_{-}$
  the probability that $u\in S$ is at most $k_{-}/n_j$. Otherwise, $u'\not\in S$, and then
  the second probability is
  \[
  \Expect{G\sim G(\vec{n}-\vec{T},P)}{Z_u = 1}
  \]
  where $\vec{T}$ is the type statistics of $S$. Using the natural coupling between $G(\vec{n},P)$
  and $G(\vec{n}-\vec{T},P)$, this probability is
  $\Expect{G\sim G(\vec{n},P)}{Z_u = 1} + \frac{k_{-}^2}{n_j} = \Expect{}{Z_u} + \frac{k_{-}^2}{n_j}$, and we are done.
\end{proof}

We are now ready to show that for each $i\in [r]$, with probability $1-o(1)$ the size of $W\cap V_i$ is $(\alpha_i + o(1)) n_i$.
\begin{lemma}\label{lem:size_of_CC_in_Vi}
  For all $i\in [r]$, we have
  $\Prob{G\sim G(\vec{n},P)}{\card{\card{W\cap V_i} - \alpha_i n_i} \leq o(n_i)} = 1-o(1)$.
\end{lemma}
\begin{proof}
  Fix $i\in [r]$, and note that $\card{W\cap V_i} = \sum\limits_{v\in V_i} Z_v$.
  Set $\zeta = n^{-0.01}$; by Chebyshev's inequality
  \[
  \Prob{}{\card{\card{W\cap V_i} - \Expect{}{\card{W\cap V_i}}}\geq \zeta n_i}
  \leq \frac{{\sf var}(\card{W\cap V_i})}{\zeta^2 n_i^2}.
  \]
  To compute the variance, we write
  \[
  {\sf var}(\card{W\cap V_i})
  =\sum\limits_{v}{{\sf var}(Z_v)}
  +\sum\limits_{v\neq v'} {\sf cov}(Z_v, Z_{v'})
  \leq n_i + n_i^2 \cdot \frac{2k_{-}^2}{n_j},
  \]
  where we used Claim~\ref{claim:bd_cov} in the last inequality.
  Thus, ${\sf var}(\card{W\cap V_i}) = \leq 4 A k_{-}^2 n$, so we get that
  \[
  \Prob{}{\card{\card{W\cap V_i} - \Expect{}{\card{W\cap V_i}}\geq \zeta n_i}}\leq \frac{4 A k_{-}^2 n}{\zeta^2 n_i^2}
  \leq \frac{4A^2 k_-^2}{\zeta^2} n^{-1}\leq \frac{1}{\sqrt{n}}
  \]
  for large enough $n$. The claim now follows as
  $\Expect{}{\card{W\cap V_i}} = n_i-n_i\cdot\Expect{}{Z_v}$ where $v\in V_i$ is some vertex, and
  by Claims~\ref{claim:expectation_compute1},~\eqref{claim:expectation_compute2}
  we have $\Expect{}{Z_v} = 1-\alpha_i + o(1)$.
\end{proof}

Applying the union bound on Lemma~\ref{lem:size_of_CC_in_Vi} over all $i\in [r]$ gives that
$\card{W\cap V_i} = (\alpha_i + o(1))n_i$ for all $i\in [r]$ except for probability $o(1)$.
Denote this event by $E_2$.

The event in question of the second item in Theorem~\ref{thm:main} is simply $E\cap E_2$, and
as each one of these events has probability $1-o(1)$, it follows that the probability of
$E\cap E_2$ is also $1-o(1)$, and we are done.\qed

\section{Numerical simulations}\label{sec:simulations}

In this section we perform several numerical simulations to highlight some
of the features of the pandemic spread in our proposed model.
In Section~\ref{sec:sim1} we perform numerical simulations and calculate the threshold for having a GCC and its size in two-type models and show the agreement
with Theorem~\ref{thm:main} (and the size of the component given in Equation~\eqref{eq:gen_fn_def} and Theorem~\ref{thm:extinct_solve_gen}).

Next, we would like to consider questions regarding the pandemic while it still develops (and not only about the end result of it).
As we proved, the basic reproduction number at the beginning of the disease spread is the
Perron-Frobenius eigenvalue $\rho(M)$ of the initial matrix $M$ and it determines whether or not the outbreak occurs.
As the spread progresses, a non-negligible fraction of the population gets infected and the ratios between
the numbers of unsaturated vertices (vertices that are susceptible) from each of the types get changed.
This follows simply from the fact there are types that are more infectious and susceptible than others and they tend to have a larger percentage of infected vertices.
Thus, while the probabilities matrix $P$ is constant throughout the spread, the matrix $M$ changes and
at each time step is modified. This change in time is not important if we are interested only in the final size of the GCC at the end of the pandemic. However, if we are interested in answering questions before the end of the pandemic, such as the GCC at the herd immunity point, we have to recalculate it at
each time step.

This calculation gives  $M_t=P\cdot {\sf diag}\left(n_1^{(t)},\dots,n_r^{(t)}\right)$, according to the amounts of unsaturated vertices $n_i^{(t)}$ (from each of the types) in the graph at time $t$. We use the SIR framework where saturated vertices are removed.
At each point in time $t$, we calculate the matrix $M_t$ and get a new Perron-Frobenius eigenvalue and a corresponding eigenvector. By updating the matrix we can follow the development
of the disease in time.
Without any countermeasures or external interventions, this time-updating eigenvalue is naturally reduced until the pandemic reaches the entire GCC, and dies out on its own. When trying to curb an outbreak in a given population, the time dependent eigenvalue can be very important, e.g. to check the possible effectiveness of countermeasures.

In Section~\ref{sec:sim2} we use our analytical results to calculate the GCC size at the end of the pandemic
and compare it to its size at the herd immunity point where the basic reproduction number
drops below one.
In Section~\ref{sec:sim3}
we calculate the GCC size at the end of the pandemic
and compare it to its size at the herd immunity point when counter measures are being taken.
We will see that countermeasures can lower the difference between the number of infected at the herd immunity
point and the end of the disease.
In Section~\ref{sec:sim4} we follow the development of the pandemic in time and show the time dependence of the
Perron-Frobenius eigenvalue and the corresponding eigenvector.

\subsection{Comparison of GCC: analytical calculation versus numerical simulation}\label{sec:sim1}

In this subsection we compare the calculation of the threshold to having GCC and its
size as follows from Theorem~\ref{thm:main} and a numerical simulation. We consider a two-type model
and sample graphs with $n_1=n_2=10,000$ vertices. We  generate a random two-type graphs with different configurations for $\lambda_{1,1}$, $\lambda_{2,2}$ and $\lambda_{1,2}=\lambda_{2,1}$.
The Perron-Frobenius eigenvalue $\rho(M)$ of the matrix $M$ reads:
\begin{equation}
 \rho(M) = \frac{M_{1,1}+M_{2,2}}{2}+ \frac{1}{2}\sqrt{(M_{1,1}+M_{2,2})^2 + 4\cdot (M_{1,2}M_{2,1}-M_{1,1}M_{2,2})} \ .
 \label{PF}
\end{equation}
In Figure~\ref{fig:realvsnum} we compare the GCC  size and the phase transition threshold
of the simulations and the analytical calculation. It is clear that, the the random graphs sampled results match the analytical results both for the threshold value of $\rho(M)$ which is $1$ and the size of the GCC.

\begin{figure}[h]
\centering
\begin{tabular}{cc}
  \includegraphics[width=75mm]{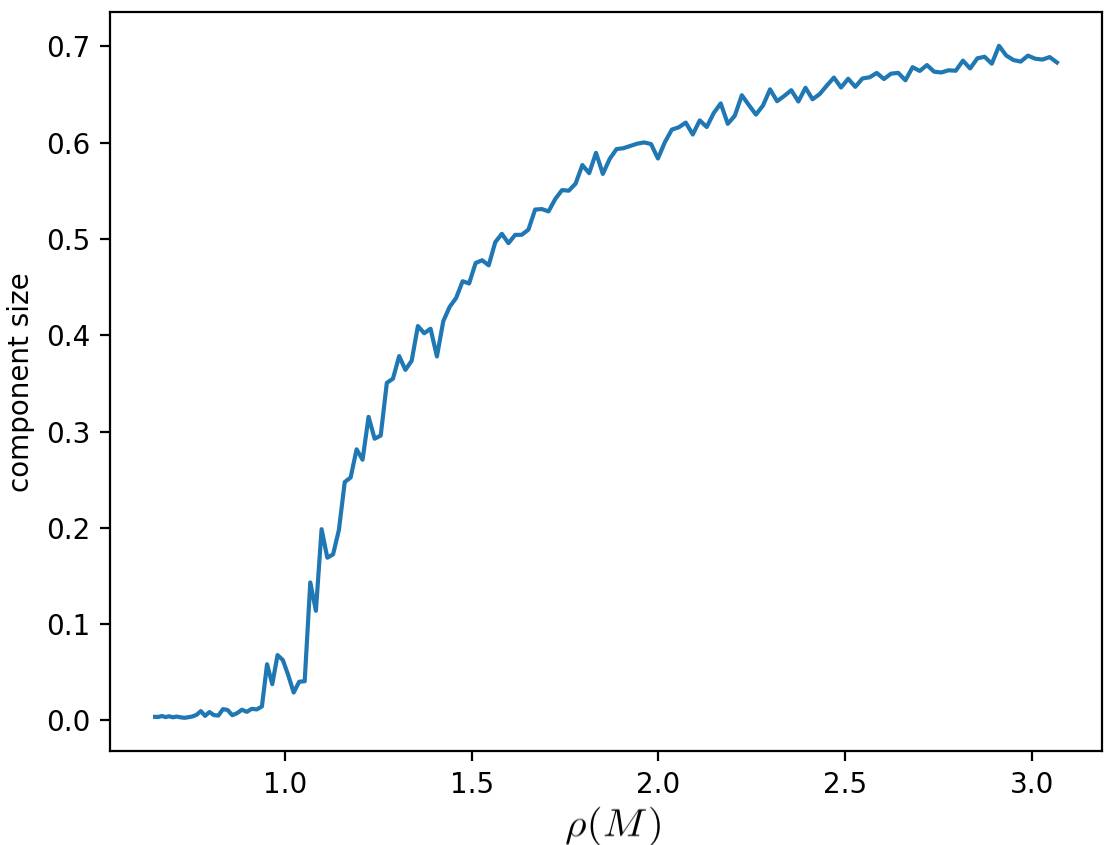} &
  \includegraphics[width=75mm]{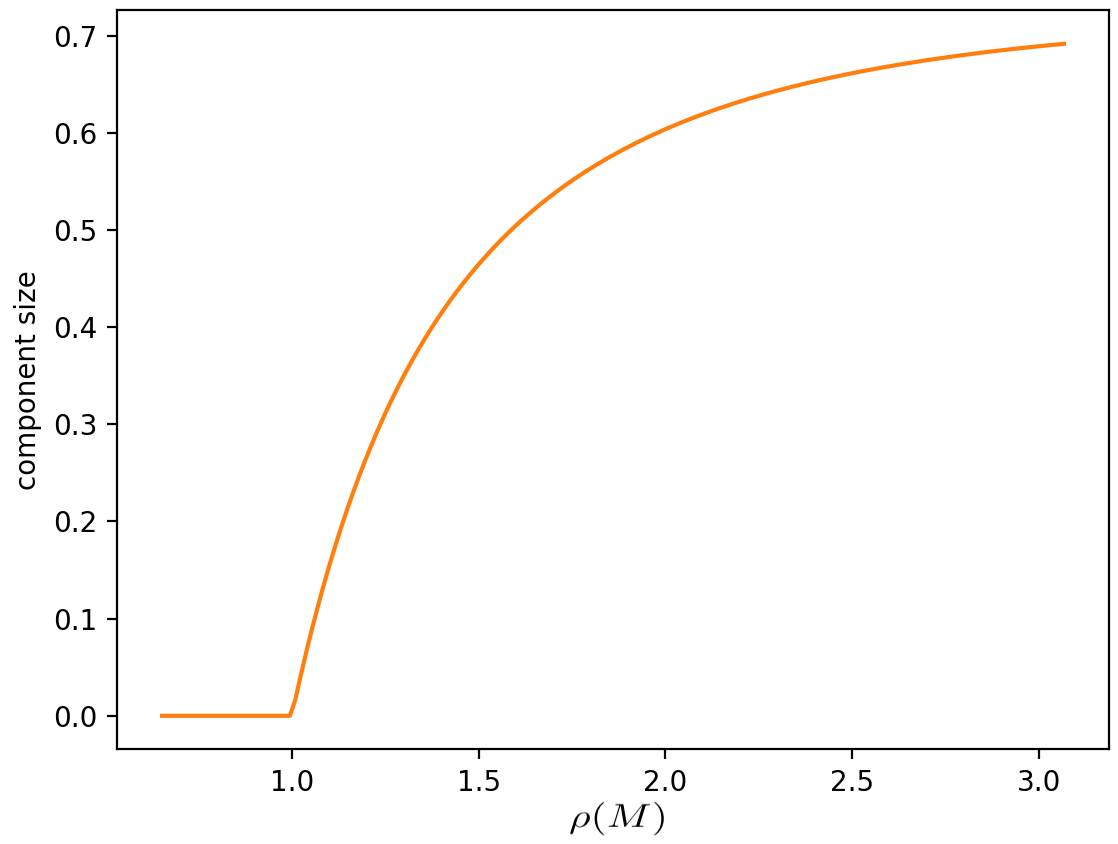}\\
  (a) & (b)\\[6pt]
\end{tabular}
\caption{(a) GCC threshold and size from the random graphs realized in the simulations. (b) GCC threshold and size derived from analytical result of Theorem~\ref{thm:main} and  Remark~\ref{remark:gcc-equation}. We use Equation~\eqref{PF} for the Perron-Frobenius eigenvalue in the two-type case. In both simulations we fixed $\lambda_{1,1}=0.2$ and $\lambda_{1,2}=0.5$ and we took values of $\lambda_{2,2}$ in the range $[0.1,3]$ with steps of $0.2$. Then $\rho(M)$ ranges in $[0,3]$}
\label{fig:realvsnum}
\end{figure}

\subsection{Herd immunity versus end of the pandemic}\label{sec:sim2}

In this subsection we consider the difference between homogeneous and heterogeneous infection graphs and the effect on the fraction of infected population at the end
of the disease. We show the difference between fraction of the population
infected at the herd immunity point where the effective reproduction number drops below one and the fraction infected at the end of the disease. This is the after-burn effect analysed in \cite{us2}.

Consider two types of populations.
Let the initial basic reproduction number be $\rho(M)=2.6$ which is in the estimated
range for the COVID-19, $R_0 \sim 2 - 3$.
%$\rho(M)$ is the Perron-Frobenius eigenvalue of the matrix $M$ and in the two-dimensional
%case reads:
%$$
% \rho(M) = \frac{M_{1,1}+M_{2,2}}{2}+ \frac{1}{2}\sqrt{(M_{1,1}+M_{2,2})^2 + 4\cdot %(M_{1,2}M_{2,1}-M_{1,1}M_{2,2})}.
%$$
We fix $\lambda_{1,2}=0.6$
and vary the ratio $\frac{\lambda_{2,2}}{\lambda_{1,1}}$ and the fraction of type one of the total population $\frac{n_1}{n_1+n_2}$.
In Figure~\ref{fig:herd} we compare
the fraction of the population infected at the end of the disease and at the herd immunity point where the effective reproduction number $\rho(M)$
drops below one (this is the standard definition, see e.g. \cite{britton2020disease,cacciapaglia2021effective}). The uppermost graphs (brown, red and orange) show the fraction of infected at the end of disease, and the lower graphs (pink, purple and green) show the fraction of infected when herd immunity is reached, i.e. $\rho(M)=1$.
We also see the difference between the two fractions of infected population in the two-dimensional
plot (b).

The difference between these two fractions is significant \cite{us2},
 is of much importance and is often ignored
in the discussions about reaching herd immunity.
We also see in the graphs the difference between the homogeneous and heterogeneous cases. Real world pandemic spread follows a heterogeneous network and we see that the fraction of infected population can be significantly lower compared to the often quoted number of the homogeneous spread.

\begin{figure}[h]
\centering
\begin{tabular}{cc}
  \includegraphics[width=60mm]{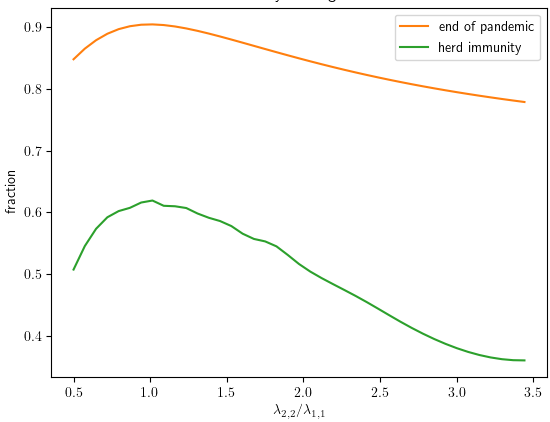} &
  \includegraphics[width=60mm]{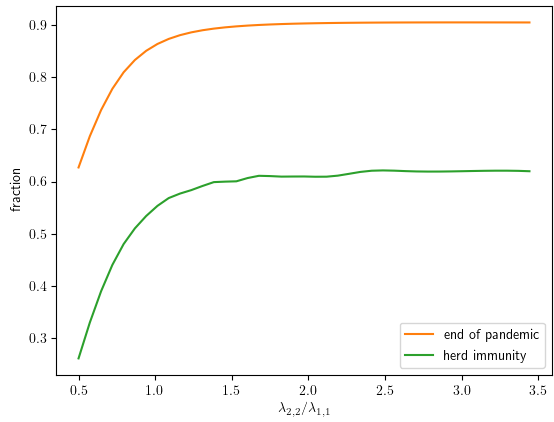}\\
  (a) & (b)\\[6pt]
\end{tabular}
\caption{A comparison between the fraction of the population infected at the end of the disease and at the herd immunity point where the effective reproduction number $\rho(M)$
drops below one. The uppermost graphs (brown, red and orange) show the fraction of infected at the end of disease, and the lower graphs (pink, purple and green) show the fraction of infected when herd immunity is reached $\rho(M)=1$.
We see a significant difference, the after-burn effect \cite{us2}.
We also see the difference between the homogeneous and heterogeneous cases, where the fraction of infected population is lower in the latter. The plot are for two types, where the initial basic reproduction number is  $\rho(M)=2.6$ and $\lambda_{1,2}=0.6$ (we move $\lambda_{1,1}$ on the range of $[0.7,3]$ with steps of $0.1$ and interpolate $\lambda_{2,2}$). In (a) and the sizes of the communities are equal, i.e. $n_1=n_2$. (b) The same effect for the same parameters only changing community sizes, i.e. $\frac{n_1}{n_1+n_2}=10\%$, i.e. the vertices of type $1$ make $10\%$ of all vertices.}
\label{fig:herd}
\end{figure}

\newpage
\subsection{The countermeasure effect}\label{sec:sim3}

We consider the effect of a partial lockdown (which stops
infections between distinct types) on the infected fraction of population at the
end of the pandemic.
The lockdown starts once
reaching the point where $10\%$ of the population is infected.
We take as an example the case of two types
with initial basic reproduction number $\rho(M)=2.6$. We use $\lambda_{1,2}=0.6$
and vary the ratio $\frac{\lambda_{2,2}}{\lambda_{1,1}}$ and the fraction of type one of the total population
 $\frac{n_1}{n_1+n_2}$.
The results are plotted in Figure~\ref{fig:ldown}.
The uppermost graphs (brown, red and orange) show the fraction of infected population at the end of disease when no countermeasure are taken. The lower graphs (pink, purple and green) show the fraction of the population that is infected if we impose the lockdown.

\begin{figure}[h]
\centering
\begin{tabular}{cc}
\includegraphics[width=60mm]{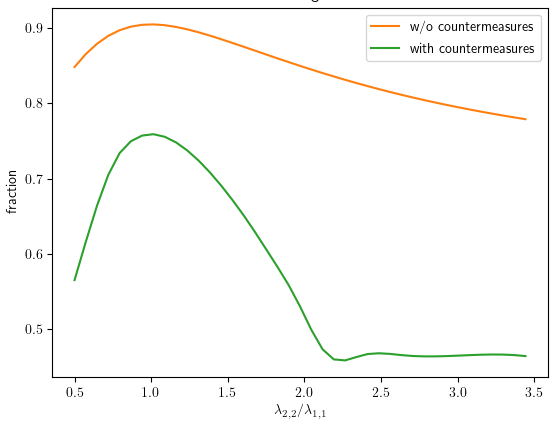} & \includegraphics[width=60mm]{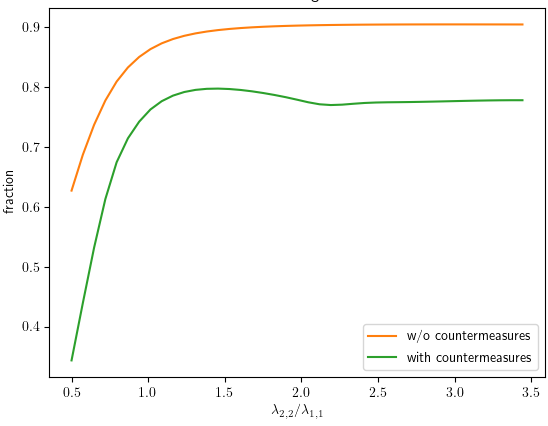}\\
(a) & (b)\\[6pt]
\end{tabular}
\caption{The effect of lockdown on the fraction of the population infected at the end of the pandemic. The uppermost graphs show fraction of infected population at the end of disease when no countermeasure are taken. The lower graphs show the fraction of the population that is infected if we impose a lockdown when reaching the point where $10\%$ of the population is infected. The lockdown is such that different types cannot infect each other and it is done by nullifying the off-diagonal entries of the matrix $\Lambda$. We consider two types, and
all the plots are drawn with initial basic reproduction number $\rho(M)=2.6$ and
$\lambda_{1,2}=0.6$, we move $\lambda_{1,1}$ on the range of $[0.7,3]$ with steps of $0.1$ and interpolate $\lambda_{2,2}$. In (a) we have the Homogeneous case where $n_1=n_2$ (b) The same effect where $\frac{n_1}{n_1+n_2}=10\%$, i.e. the vertices of type $1$ make $10\%$ of all vertices.}
\label{fig:ldown}
\end{figure}

\subsection{The direction of the disease spread}\label{sec:sim4}

One of our observations is that despite the complex structure of the pandemic spread
one can identify a propagation direction. It is given by the Perron-Frobenius eigenvector
of the matrix $M$, where the corresponding eigenvalue is the reproduction
number of the disease whose threshold value at $1$ separates between an outbreak
or no outbreak of the pandemic. We refer to the Perron-Frobenius eigenvector which is time dependent (i.e. of the matrices $M_t$ defined above) and in order to calculate it one needs to update the matrix $M$, to get $M_t$, as the pandemic evolves.
Intuitively, the Perron-Frobenius eigenvector points at any time step in the direction
where there is the highest potential to infect. This takes into account the size
of the remaining uninfected population of the different types at that time step and
their probability to infect. In fact, it supports the result of Theorem $1$ of ~\cite{MeaningEV}.

In Figure~\ref{fig:evec} we plot the Perron-Frobenius eigenvector for two distinct scenarios: homogeneous and
heterogeneous infection graphs.
The projections on the three axes show the weight of the three different types
in the pandemic propagation. In plot (a) we see the homogeneous case, where all the three types are equal in size, $\rho(M)\approx 2.1$ and we use:

\[
\Lambda=
  \begin{pmatrix}
    1.47 & 0.3 & 0.3 \\
    0.3 & 1.47 & 0.3 \\
    0.3 & 0.3 & 1.47
  \end{pmatrix}
  \label{L1}
\]
\vskip0.2cm
\noindent
As intuitively expected, the direction of the
propagation is time independent.

In plots (b) and (c) we see the heterogeneous case. In this example we take
 $80\%$ from type one and $10\%$ from each of types $2$ and $3$. $\rho(M)\approx 2.7$ and we used
\[
\Lambda=
  \begin{pmatrix}
    1.16 & 0.18 & 0.18 \\
    0.18 & 2.32 & 0.36 \\
    0.18 & 0.36 & 2.32
  \end{pmatrix}
  \label{L2}
\]
\vskip0.2cm

\noindent Types $2$ and $3$ are relatively highly infectious, and more susceptible to be infected than type $1$. This means that they are quickly being infected and after they infect others they are removed (become saturated vertices). However, Type $1$ is much less infectious and susceptible to be infected and remains longer with more unsaturated vertices, this means that only at a later stage of the outbreak, when vertices of types $2$ and $3$ become saturated, the pandemic catches more vertices of type $1$ much quicker. We see in the plot that types $2$ and $3$ decrease quickly  in their potential to infect compared to type $1$ (their amount of unsaturated vertices is decreased), while type $1$'s unsaturated vertices increase relative to types $2$ and $3$ as the pandemic evolves.

\begin{figure}[h]
\centering
\begin{tabular}{ccc}
 \includegraphics[width=50mm]{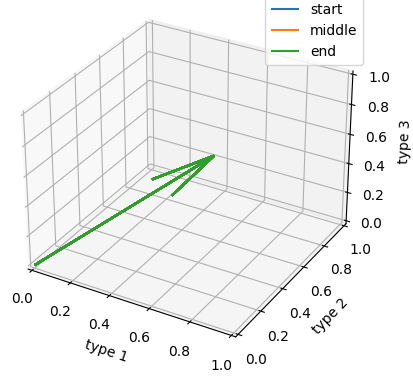} & \includegraphics[width=50mm]{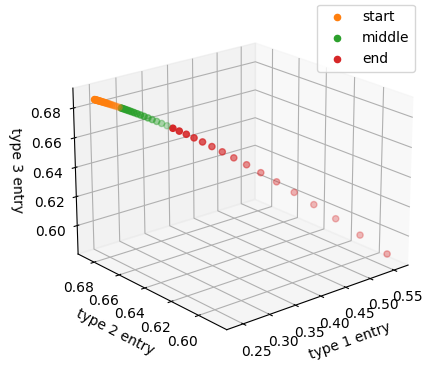} & \includegraphics[width=50mm]{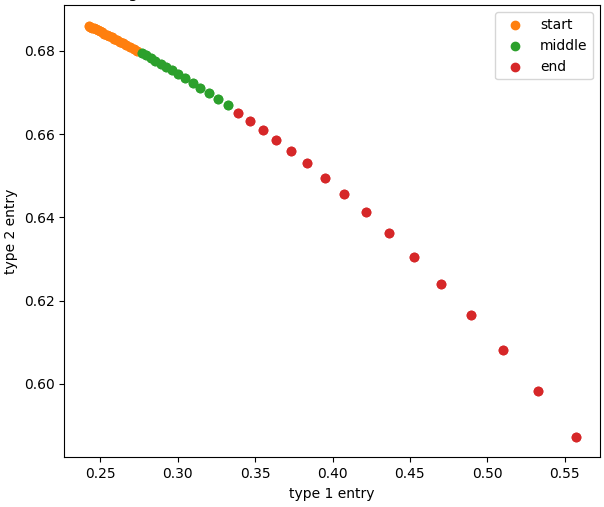}\\
(a) & (b) & (c)\\[6pt]
\end{tabular}
\caption{(a) The position of the Perron-Frobenius eigenvector during the evolution of the pandemic in the homogeneous spread case (constant and centered). All three types are equal in size and $\rho(M)\approx 2.1$. (b) and (c - projection of (b) on types $1$ and $2$)
The position of the Perron-Frobenius eigenvector during the evolution of the pandemic in the non-homogeneous spread case viewed from different perspectives. The entries of types $2$ and $3$ are decreased and the entry of type $1$ is increased, while the outbreak progresses.}
\label{fig:evec}
\end{figure}

\newpage

\section{Discussion}\label{sec:discussion}

The contact graph of a real world pandemic is naturally heterogeneous and complex. It is clearly desirable to be able to work with a general graph and this is what we have done in this work. We employed the multi-type Galton-Waston branching process and analysed the question of whether there would be an outbreak and what will be the fraction of infected population at the end of the disease, in the case of an outbreak.

We defined an $(r\times r)$-dimensional matrix
$M$ where $r$ is the number of types, and whose entries encode the probability
that an individual of one type would infect an individual of his type or an individual from another type (and this for each pair of types).
$M$ encodes the whole information about the evolution of the pandemic. In particular, its
Perron-Frobenius eigenvalue is the basic reproduction number that determines
whether there will be an outbreak.
The corresponding eigenvalue points in the the direction of the spread and
at each time step takes into account the remaining individuals of each type that can infect as well as their probability to infect.

Our framework allows for a general simulation on real world data once collected.
It would be of importance to follow this direction.
In our numerical simulations we presented several examples that highlight certain
properties of the spread such as the difference between homogeneous and heterogeneous
infection networks and the difference between herd immunity and the final end of the disease spread. We have also observed a property of our model regarding the eigenvector corresponding to the Perron-Frobenius eigenvalue. We have shown that, at least numerically, this eigenvector describes the direction in which the pandemic will spread (i.e. to which types first and quicker and to which types later on).

Further developments of our model may include applying it to more specific cases with explicit types, as done in~\cite{NewmanMed}. One may want to use our model in order to incorporate connections between caregivers (may be viewed as super-spreaders of one type) and patients (that may be themselves partitioned into different communities). Through this, one can study the effects of more subtle interventions and countermeasures, or plan effective assignment of caregivers, in order to prevent further infections and stop the outbreak. Also, it is to be considered that other frameworks, rather than SIR, are more applicable in some cases, as described in~\cite{Karrer_2010}, e.g. by allowing non immediate recovery-time (and thus prolonging infectiousness stage), but rather letting it follow a certain distribution.

Moreover, our analysis is applicable for a general information spread and as such outlines a structure that can be used not just for a pandemic spread.
A generalization of our work that is worth pursuing in this context is to a non-symmetric probability
matrix $P$ that naturally arises, for instance, in search engines.
Another important direction to follow is to have general degree distributions, rather than just binomial, e.g.
the Gamma distribution that has proven valuable is describing pandemics such as COVID-19.

\bibliography{ref.bib}
\end{document}